\documentclass[sigconf]{acmart}
\AtBeginDocument{%
  }

\setcopyright{acmlicensed}
\copyrightyear{2025}
\acmYear{2025}
\setcopyright{rightsretained}
\acmConference[CIKM '25]{Proceedings of the 34th ACM International
	Conference on Information and Knowledge Management}{November 10--14,
	2025}{Seoul, Republic of Korea}
\acmBooktitle{Proceedings of the 34th ACM International Conference on
	Information and Knowledge Management (CIKM '25), November 10--14, 2025,
	Seoul, Republic of Korea}\acmDOI{10.1145/3746252.3761067}
\acmISBN{979-8-4007-2040-6/2025/11}

\usepackage{hyperref}
\usepackage{graphicx}
\usepackage{textgreek}
\usepackage{algorithm}
\usepackage{tikz}
\usepackage{amsthm}
\usepackage{mdwlist}

\usepackage{xspace}
\usepackage{amsmath}
\usepackage{algorithmic}
\usepackage{mdwlist}
\usepackage{color}
\usepackage{csquotes}
\usepackage{booktabs}
\usepackage{multirow}
\usepackage{bbm}
\usepackage{bm}
\usepackage{xcolor}
\usepackage{threeparttable}
\usepackage{graphicx}
\usepackage{subcaption}
\usepackage{url}
\usepackage{enumitem}
\setlist{nosep}

\newcommand{\method}{ReVol\xspace}
\newcommand{\methodlong}{\underline{Re}turn-\underline{Vol}atility Normalization for Mitigating Distribution Shift in Stock Price Data\xspace}
\newcommand{\RVN}{$RVN$\xspace}
\newcommand{\RVE}{$RVE$\xspace}
\newcommand{\RVD}{$RVD$\xspace}

\newtheorem{problem}{Problem}

\begin{document}
\settopmatter{printacmref=false} 
\renewcommand\footnotetextcopyrightpermission[1]{}
	
\setlength{\floatsep}{0.01cm}
\setlength{\textfloatsep}{0.1cm}
\setlength{\intextsep}{0.01cm}
\setlength{\dblfloatsep}{0.01cm}
\setlength{\dbltextfloatsep}{0.01cm}
\setlength{\abovedisplayskip}{0.01cm}
\setlength{\belowdisplayskip}{0.01cm}
\setlength{\abovecaptionskip}{0.01cm}
\setlength{\belowcaptionskip}{0.01cm}

\title{Mitigating Distribution Shift in Stock Price Data via Return-Volatility Normalization for Accurate Prediction}

\author{Hyunwoo Lee}
\affiliation{%
	\institution{Seoul National University}
	\city{Seoul}
	\country{South Korea}
}
\email{jwarhw98@snu.ac.kr}
\author{Jihyeong Jeon}
\affiliation{%
	\institution{Seoul National University}
	\city{Seoul}
	\country{South Korea}
}
\email{jeonjihyeong@snu.ac.kr}
\author{Jaemin Hong}
\affiliation{%
	\institution{Seoul National University}
	\city{Seoul}
	\country{South Korea}
}
\email{jmhong0120@snu.ac.kr}
\author{U Kang}
\affiliation{%
	\institution{Seoul National University}
	\city{Seoul}
	\country{South Korea}
}
\email{ukang@snu.ac.kr}
%\email{jwarhw98@snu.ac.kr, jeonjihyeong@snu.ac.kr, jmhong0120@snu.ac.kr, ukang@snu.ac.kr}

%\author{Hyunwoo Lee}
%\affiliation{%
%	\institution{Seoul National University}
%	\city{Seoul}
%	\country{South Korea}
%}
%\email{jwarhw98@snu.ac.kr}
%
%\author{Jihyeong Jeon}
%\affiliation{%
%	\institution{Seoul National University}
%	\city{Seoul}
%	\country{South Korea}
%}
%\email{jeonjihyeong@snu.ac.kr}
%
%\author{Jaemin Hong}
%\affiliation{%
%	\institution{Seoul National University}
%	\city{Seoul}
%	\country{South Korea}
%}
%\email{jmhong0120@snu.ac.kr}
%
%\author{U Kang}
%\affiliation{%
%	\institution{Seoul National University}
%	\city{Seoul}
%	\country{South Korea}
%}
%\email{ukang@snu.ac.kr}

%\author{Jihyeong Jeon}
%\affiliation{%
%	\institution{Seoul National University}
%	\city{Seoul}
%	\country{South Korea}
%}
%\email{jeonjihyeong@snu.ac.kr}
%
%\author{Jaemin Hong}
%\affiliation{%
%	\institution{Seoul National University}
%	\city{Seoul}
%	\country{South Korea}
%}
%\email{jmhong0120@snu.ac.kr}
%
%\author{U Kang}
%\affiliation{%
%	\institution{Seoul National University}
%	\city{Seoul}
%	\country{South Korea}
%}
%\email{ukang@snu.ac.kr}


\begin{abstract}
How can we address distribution shifts in stock price data to improve stock price prediction accuracy?
%Stock price prediction technology has garnered significant attention in both academia and industry due to its ability to recognize complex patterns.
%Existing methods often fail to address distribution shift problem effectively, assuming stationary price patterns and focusing only on price scale normalization, while ignoring critical factors like return and volatility.
Stock price prediction has attracted attention from both academia and industry, driven by its potential to uncover complex market patterns and enhance decision-making. However, existing methods often fail to handle distribution shifts effectively, focusing on scaling or representation adaptation without fully addressing distributional discrepancies and shape misalignments between training and test data.

We propose \method(\methodlong), a robust method for stock price prediction that explicitly addresses the distribution shift problem.
\method leverages three key strategies to mitigate these shifts: (1) normalizing price features to remove sample-specific characteristics, including return, volatility, and price scale, (2) employing an attention-based module to estimate these characteristics accurately, thereby reducing the influence of market anomalies, and (3) reintegrating the sample characteristics into the predictive process, restoring the traits lost during normalization.
Additionally, \method combines geometric Brownian motion for long-term trend modeling with neural networks for short-term pattern recognition, unifying their complementary strengths.
Extensive experiments on real-world datasets demonstrate that \method enhances the performance of the state-of-the-art backbone models in most cases, achieving an average improvement of more than 0.03 in IC and over 0.7 in SR across various settings. %These results underscore the importance of directly confronting distribution shifts and highlight the efficacy of our integrated approach for stock price prediction.
\end{abstract}

\maketitle

\section{Introduction}

% \begin{figure}
% 	\centering
% 	\begin{minipage}{0.98\linewidth}
% 		\centering
% 		\begin{subfigure}{0.49\textwidth}
% 			\includegraphics[width=\textwidth]{figures/KO.png}
% 			\caption{KO stock price} \label{fig:KO}
% 		\end{subfigure} \hfill
% 		\begin{subfigure}{0.49\textwidth}
% 			\includegraphics[width=\textwidth]{figures/raw_hist.pdf}
% 			\caption{Raw price data distributions} \label{fig:raw_hist}
% 		\end{subfigure} \hfill
% 		\begin{subfigure}{0.49\textwidth}
% 			\includegraphics[width=\textwidth]{figures/dishts_hist.pdf}
% 			\caption{Distributions of normalized data from Dish-TS~\cite{fan2023dish}} \label{fig:dishts_hist}
% 		\end{subfigure} \hfill
% 		\begin{subfigure}{0.49\textwidth}
% 			\includegraphics[width=\textwidth]{figures/revol_hist.pdf}
% 			\caption{Distributions of normalized data from \method} \label{fig:revol_hist}
% 		\end{subfigure}
% 	\end{minipage}
% 	\caption{(a) shows a stock price dataset, including both training and test periods.
%     (b) depicts distribution difference between raw price data of the training and test periods.
%     (c) Dish-TS~\cite{fan2023dish} reduces the distribution difference; however, it does not fully align the distributions.
%     (d) \method aligns the two distributions near perfectly.
%     }
%     \label{figure:dist}
% \end{figure}

\begin{figure}[t]
	\centering
    \includegraphics[width=0.48\textwidth]{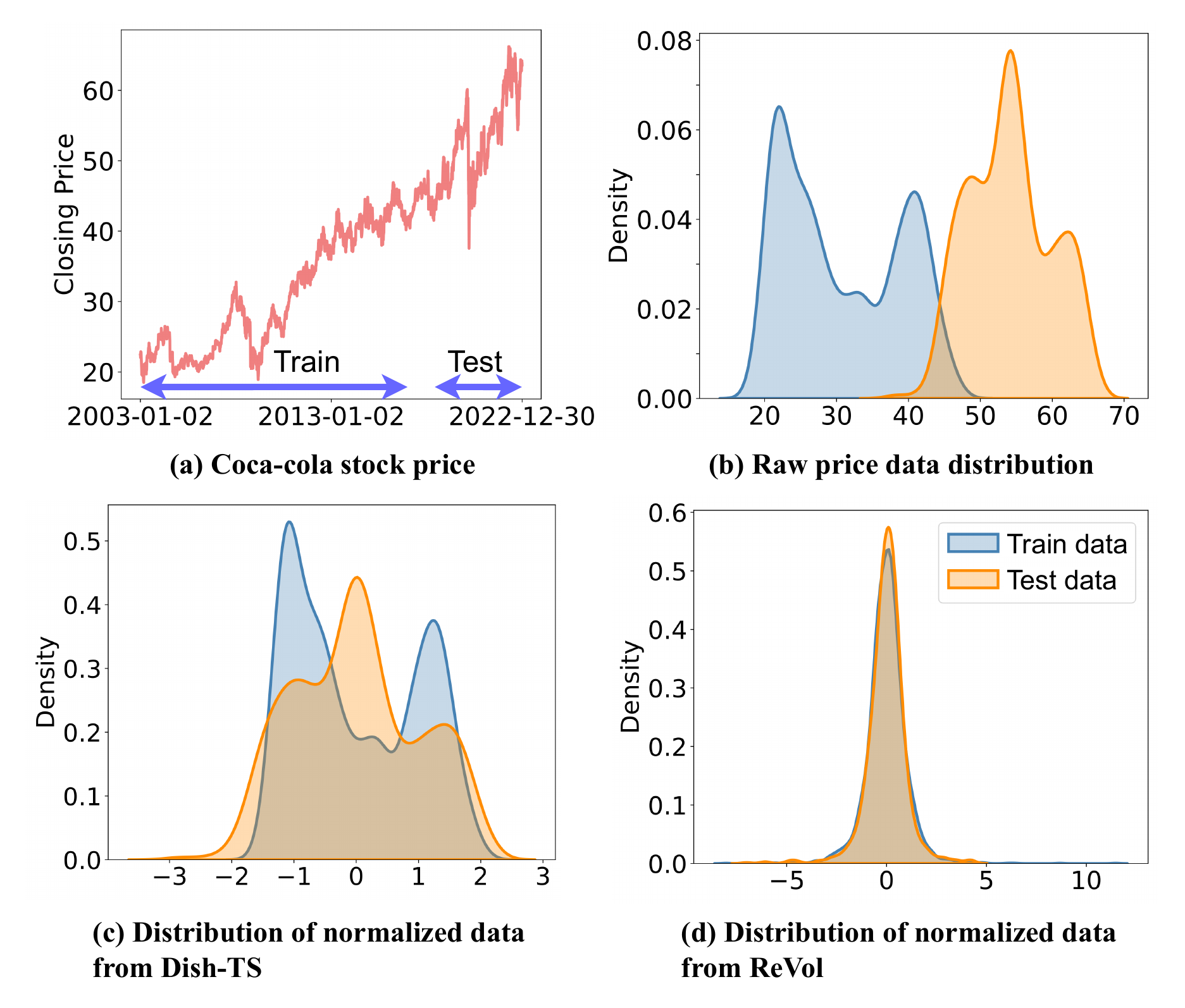}
	\caption{(a) shows a stock price dataset, including both training and test periods.
    (b) depicts distribution difference between raw price data of the training and test periods.
    (c) Dish-TS~\cite{fan2023dish} reduces the distribution difference; however, it does not fully align the distributions.
    (d) \method effectively aligns the two distributions with minimal discrepancy.}
    \label{figure:dist}
\end{figure}

How can we address distribution shifts in stock price data to improve predictive accuracy?
Distribution shifts --- distribution discrepancies across samples --- pose a fundamental challenge in data mining and machine learning. This issue is especially critical in stock price prediction, where market conditions evolve continually due to economic factors, geopolitical events, and investor sentiment.
These shifts occur not only over time as market conditions change but also across different stocks, each exhibiting unique statistical properties such as return and volatility.
Successfully overcoming these shifts is crucial for constructing robust and trustworthy predictive models, which in turn support more reliable investment strategies and enhance financial outcomes.

Despite extensive work on stock price prediction, existing methods~\cite{feng2019enhancing,li2021modeling,yoo2021accurate} often fail to adequately deal with distribution shifts. They do not explicitly mitigate distributional differences across samples, limiting their generalization in non-stationary financial time series. Methods for addressing distribution shifts~\cite{kim2021reversible,fan2023dish} unify distribution parameters across samples, while failing to align distribution shapes. Figure~\ref{figure:dist} illustrates how raw price data and normalized data distributions differ between training and test periods. Others~\cite{lin2021learning,zhan2024meta,zhao2023doubleadapt} rely on adaptation to the representation space, increasing the complexity of the model without directly resolving the distribution shift. These limitations hinder the robustness and accuracy of the prediction under changing market conditions.

We propose \method (\methodlong), a method for robust stock price prediction that explicitly addresses distribution shifts by normalizing individual stock characteristics. The main ideas of \method are as follows. First, \method mitigates distribution shifts by removing sample-specific attributes, such as return, volatility, and price scale, from the input data. Second, it accurately estimates these characteristics through an attention-based weighted averaging technique that alleviates the impact of transient market anomalies.
Third, \method reintegrates these characteristics into a prediction pipeline to recover the information removed during normalization.

Our contributions are summarized as follows:
\begin{itemize*}
    \item \textbf{Addressing distribution shifts in stock price prediction.} We propose \method, a novel approach to address distribution shifts by normalizing individual sample characteristics, including return, volatility, and price scale.
    Figure~\ref{figure:dist} depicts that the distribution gap between training and test data is mitigated by \method.
    \item \textbf{Integration of short-term patterns and long-term trends.} \method unifies geometric Brownian motion (GBM) for modeling long-term price behavior with a neural network for capturing nonlinear short-term patterns, thereby leveraging the complementary strengths of both methods.
    \item \textbf{Experiments.} We demonstrate that \method improves the performance of the state-of-the-art backbone models in most scenarios, delivering an average gain of more than 0.03 in IC and over 0.7 in SR. Moreover, \method effectively narrows the distribution gap between training and test data via the proposed GBM-based normalization.
\end{itemize*}

\begin{table}
	\centering
	\caption{Symbols and descriptions.}
	\small
	\begin{tabular}{@{}ll@{}}
		\toprule
		\textbf{Symbol} & \textbf{Definition} \\ \midrule
		$S_t$           & Stock price at time $t$ \\
		$B_t$           & Value of Brownian motion at time $t$ \\
		$\mu$           & Return of a stock \\
		$\sigma$        & Volatility of a stock \\
		$S_t^c, S_t^o, S_t^h, S_t^l$ & Closing, opening, the highest, and the lowest prices at time \\
		& step $t$ \\
		$\epsilon_t^c, \epsilon_t^o, \epsilon_t^h, \epsilon_t^l$ & Error terms of closing, opening, the highest, and the lowest \\
		& prices at time step $t$ \\
		$T$				& The current time step \\
		$r$             & Time interval between opening and previous day's \\
		& closing prices \\
		$u$             & Time interval between the highest and previous day's \\
		& closing prices \\
		$v$             & Time interval between the lowest and previous day's \\
		& closing prices \\
		$\mathbf{x}_t$	& Price feature vector at time step $t$ \\
		$\tilde{\mathbf{x}}_t$	& Feature vector at time step $t$ from fully connected layer \\
		$\mathbf{h}_t$	& Hidden state vector at time step $t$ from LSTM \\
		$\alpha_t$		& Attention weight at time step $t$ \\
		$\theta$        & Parameters of the neural network backbone \\
		$f_\theta$      & Neural network backbone \\
		$\phi$   & Parameters of the return-volatility estimator module \\
		$\RVE_\phi$      & Return-volatility estimator module \\
		$\RVN$   & Return-volatility normalization module \\
		$w$             & The window size \\
		$\beta$			& Hyperparameter to control the weight of guidance loss \\ \bottomrule
	\end{tabular}
	\label{tab:symbols}
\end{table}

%The rest of this paper is organized as follows. We introduce related research on stock price prediction in Section \ref{related}. We propose \method in Section \ref{method}. We present the experimental results on real-world stock datasets in Section \ref{experiments} and conclude in Section \ref{conclusion}.
The definitions of the symbols in this paper are listed in Table~\ref{tab:symbols}.
The code and datasets are available at \url{https://github.com/AnonymousCoder2357/ReVol}. 
\section{Related Works} \label{related}

\begin{figure*}[ht!]
	\centering
	\includegraphics[width=0.85\textwidth]{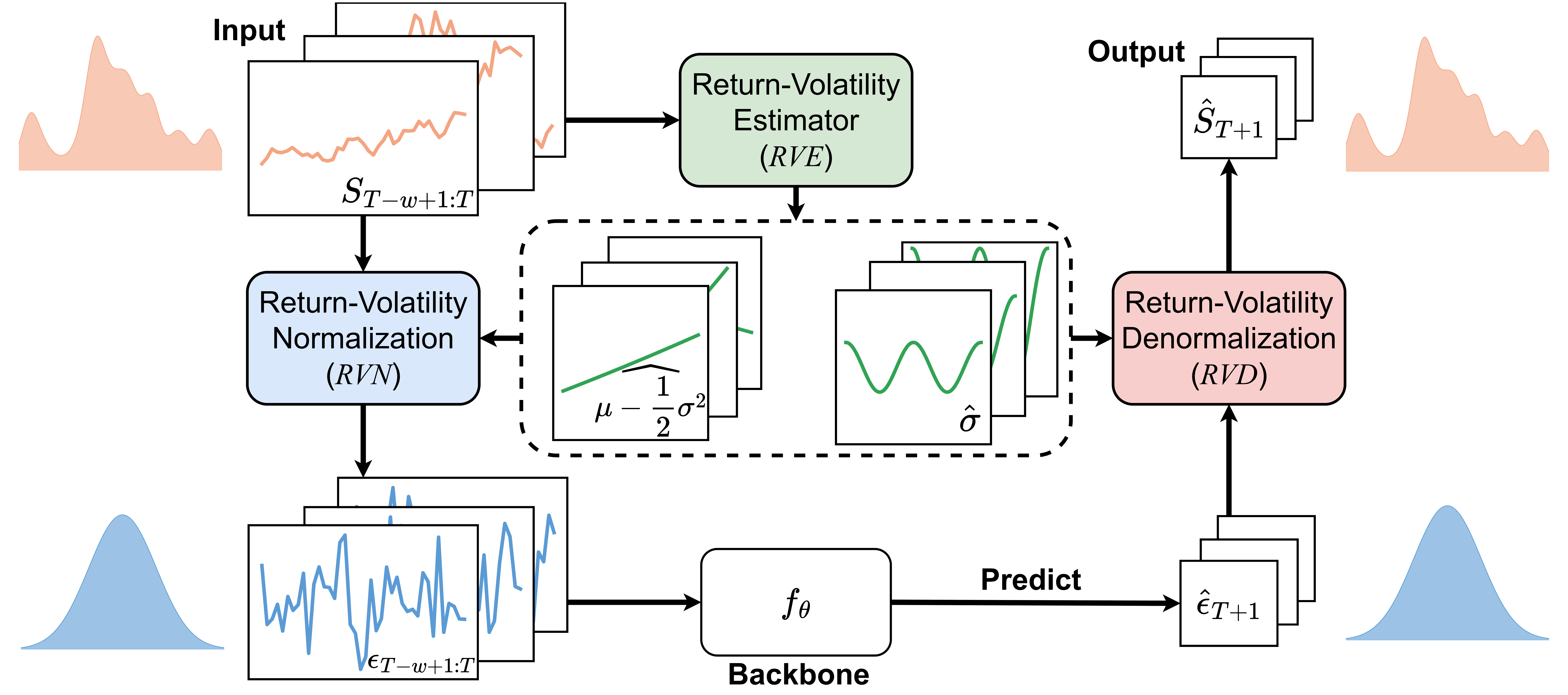}
	\caption{An overview of \method.}
	\label{fig:overview}
\end{figure*}

\subsection{Geometric Brownian Motion} \label{sec:related:gbm}
Geometric Brownian motion (GBM) \cite{bjork2019arbitrage} is one of the most popular models for stock price dynamics. GBM is widely used in various financial tasks, such as option pricing and portfolio optimization, due to its suitability and simplicity.
The GBM model assumes that the logarithm of stock prices follows a Brownian motion with drift, allowing for continuous price movements. In GBM, the stock price $S_t$ at time $t$ is modeled as:

\begin{equation} \label{gbm}
S_t = S_0 e^{\left( \mu - \frac{1}{2} \sigma^2 \right)t + \sigma B_t}
\end{equation}

\noindent where $\mu$, $\sigma$, and $B_t$ are the return, volatility of the stock, and the value of Brownian motion at time $t$, respectively. Brownian motion (BM) $B$ is a stochastic process with the following properties:

\begin{enumerate}
    \item $B_0=0$
    \item \label{p2} if $0 \leq s_1 \leq t_1 \leq s_2 \leq t_2$, then $B_{t_1}-B_{s_1}$ and $B_{t_2}-B_{s_2}$ are independent
    \item \label{p3} $B_t - B_s \sim \mathcal{N}(0, t - s) \quad \text{(for } 0 \leq s \leq t\text{)}$
    \item $t \to B_t$ is almost surely continuous
\end{enumerate}

\noindent where $B_t$ is the value of BM at time $t$.

For discrete time $t$, $S_t$ is expressed in terms of $S_{t-1}$ by representing the BM part as $\epsilon_t \sim \mathcal{N}(0, 1)$ with the following Lemma.

\begin{lemma} \label{lemma:discrete_gbm}
	When time is discretized, the stock price $S_t$ at time $t\in\mathbb{N}$ is  expressed as follows:
	\begin{equation} \label{discrete_gbm}
	S_{t} = S_{t-1} e^{\left( \mu - \frac{1}{2} \sigma^2 \right) + \sigma \epsilon_{t}}
	\end{equation}
	\noindent where the error term $\epsilon_{t} \sim \mathcal{N}(0, 1)$. For any $s, t \in \mathbb{N}$ with $s < t$, the error terms $\epsilon_s$ and $\epsilon_t$ corresponding to timepoints $s$ and $t$, respectively, are independent.
\end{lemma}

\begin{proof}
	By Equation~\eqref{gbm}, the stock price $S_{t-1}$ at time $t-1$ is given as follows:
	\begin{equation} \label{gbm_t_1}
	S_{t-1} = S_0 e^{\left( \mu - \frac{1}{2} \sigma^2 \right)(t-1) + \sigma B_{t-1}}.
	\end{equation}
	\noindent By dividing Equation~\eqref{gbm} by Equation~\eqref{gbm_t_1}, we obtain the following:
	\begin{equation} \label{gbm_t/t1}
	\frac{S_t}{S_{t-1}}=e^{\left(\mu-\frac{1}{2}\sigma^2\right)+\sigma\left(B_t-B_{t-1}\right)}.
	\end{equation}
	\noindent Let $\epsilon_t=B_t-B_{t-1}$. Then, $\epsilon_t$ and $\epsilon_s$ are independent for $s<t$ and $s\in\mathbb{N}$, and $\epsilon_{t} \sim \mathcal{N}(0, 1)$ due to properties \ref{p2} and \ref{p3} of BM, respectively.
	By replacing $B_t-B_{t-1}$ with $\epsilon_t$ in Equation~\eqref{gbm_t/t1}, we obtain the result.
\end{proof}

\noindent The assumption that $\epsilon_t$ and $\epsilon_s$ are independent for $s<t$ and $t,s\in\mathbb{N}$ means that future stock prices cannot be predicted from past prices.
In contrast, we relax this assumption to allow learning from historical dependencies.

\subsection{Deep Learning-Based Stock Price Prediction}
With the advancement of deep neural networks, attempts to predict stock prices using the time series processing capabilities of recurrent neural networks have emerged \cite{nelson2017stock, sunny2020deep}.
The attention mechanism has been used to enhance the performance of RNN-based models \cite{10.5555/3172077.3172254, li2018stock, feng2019enhancing, yoo2021attention, soun2022accurate, kim2024accurate}. GCN or Transformer architectures have been used to leverage the correlations between stocks \cite{li2021modeling, yoo2021accurate, li2024master}.
However, these methods struggle to address the distribution shift problem since they fail to alleviate distributional discrepancy between training and test data.
Instead, we focus explicitly on mitigating distribution shift rather than solely learning representations.

\subsection{Methods for Addressing Distribution Shift}
As distribution shift has emerged as a critical issue in time series analysis, various methods have been proposed to address it in the time series and financial domains.
RevIN~\cite{kim2021reversible} and Dish-TS~\cite{fan2023dish} equalize the distribution parameters, mean and standard deviation, for each sample through instance normalization.
However, financial time series data do not generally follow a normal distribution.
As a result, even if the mean and standard deviation are equalized, the distribution shapes still differ across samples due to higher-order characteristics such as skewness and kurtosis.
There have been also methods to detect distribution shifts that may occur during the evolution of stock time series, and adapt to various market conditions \cite{lin2021learning, zhao2023doubleadapt, zhan2024meta}.
However, these methods rely on the representation space to adapt to the distribution on its own, which not only increases the burden on the model but also fails to explicitly address distribution shift.
In this work, we aim to resolve the distribution shift problem by making the distribution of each sample identical to improve the accuracy of stock price prediction. 
\section{Proposed Method} \label{method}

We propose \method (\methodlong), an accurate method for stock price prediction by mitigating the distribution shift problem.

\subsection{Overview}
%We aim to accurately predict future stock prices by alleviating the distribution shift.
We give a formal definition of the stock price prediction problem as follows:

\begin{problem}[Stock Price Prediction.]
	\textbf{Given} the historical prices $S_{T-w+1:T}^o$, $S_{T-w+1:T}^h$, $S_{T-w+1:T}^l$, and $S_{T-w+1:T}^c$ from day $T-w+1$ up to $T$,  where $S_t^o$, $S_t^h$, $S_t^l$, and $S_t^c$ are opening, the highest, the lowest,and closing prices on day $t$, respectively,
	\textbf{find} the closing price $S_{T+1}^c$ for the next day.
\end{problem}

To achieve accurate prediction mitigating distribution shift, \method is designed to address the following challenges.

\begin{enumerate*}
	\item \textbf{Distribution shift in stock price data.} The distribution of stock price data varies across stocks and over time.
	This distribution shift problem prevents accurate stock price prediction.
	However, existing methods focus only on matching the mean and standard deviation across samples, leaving the distribution shape unaligned.
	How can we reduce the discrepancies between distributions?
	\item \textbf{Estimating sample characteristics.} We must accurately estimate and remove individual sample characteristics, from stock price data to reduce distribution discrepancies.
	However, due to external shocks, stock prices sometimes deviate temporarily from their true distribution.
	Such anomalies hinder the estimation of the true distribution when distribution parameters are naively estimated using arithmetic means.
	How can we estimate the true distribution accurately?
	\item \textbf{Reintegrating sample characteristics.} Removing sample characteristics from the input data interrupts accurate prediction since the model cannot properly reconstruct the data distribution.
	How can we reintegrate the individual sample characteristics into the prediction?
\end{enumerate*}

Our proposed \method addresses the challenges with the following ideas.
Figure~\ref{fig:overview} depicts the overall process of \method.

\begin{enumerate*}
	\item \textbf{Return-volatility normalization (Section \ref{rvn}).} Inspired by GBM, we normalize the input stock price data with sample characteristics (return, volatility, and price scale), to remove them from the data.
	This effectively reduces distributional discrepancies across samples.
	As a result, we obtain historical error terms, which serve as inputs to the neural network backbone.
	\item \textbf{Return-volatility estimator (Section \ref{rve}).} We calculate return and volatility using a weighted average with attention weights.
	This approach mitigates the impact of anomalies that disrupt characteristic estimation.
	Consequently, we acquire precisely estimated sample characteristics.
	\item \textbf{Backbone and Return-volatility denormalization (Section \ref{rvd}).} We predict the future error term with the neural network backbone and denormalize it with sample characteristics to obtain the final prediction.
	By doing so, we restore the sample characteristics lost during the normalization process.
\end{enumerate*}

Each of the above idea corresponds to each main module of \method.
First, \RVN (Return-Volatility Normalization) module normalizes historical stock prices with sample characteristics to generate historical error terms, which act as inputs for the neural network backbone.
The backbone model does not suffer from distribution shift since the error terms follow an identical distribution across samples.
Second, \RVE (Return-Volatility Estimator) module accurately estimates sample characteristics through an attention mechanism, which are subsequently used in \RVN and \RVD (Return-Volatility Denormalization).
Lastly, the backbone predicts the future error term, and \RVD module denormalizes it to reintegrate information removed by \RVN.
The denormalized output provides the final future price prediction.

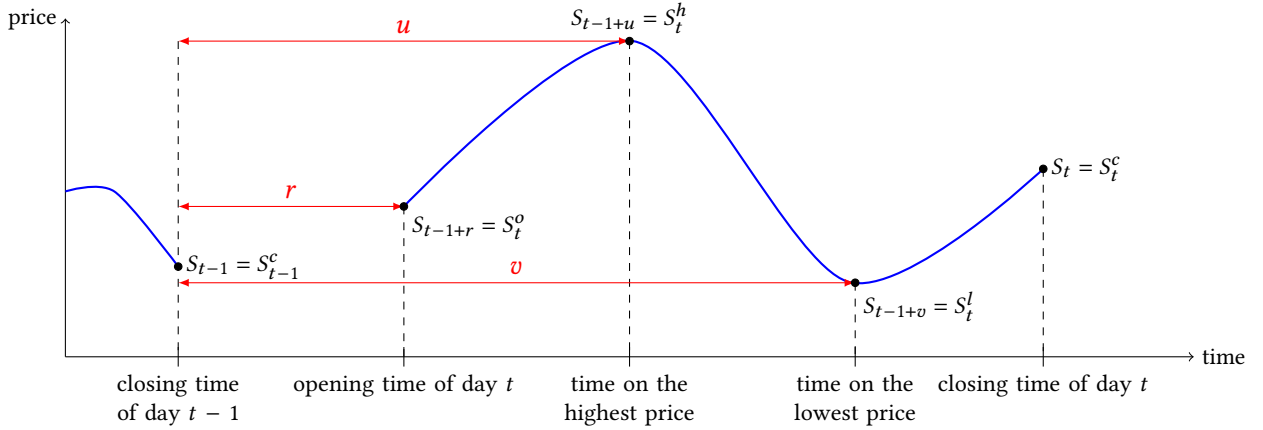
\begin{figure*}
	\begin{tikzpicture}
	
	\draw[->] (0, 0) -- (15, 0) node[anchor=west] {time};
	\draw[->] (0, 0) -- (0, 4.5) node[anchor=east] {price};
	
	\draw[thick, smooth,blue]
	plot coordinates {(0, 2.2) (0.65, 2.2)  (1.5,1.2)};
	
	\draw[thick, smooth,blue]
	plot coordinates {(4.5, 2) (7.55, 4.2) (10.4, 1) (13,2.5)};
	
	%dash, tick
	\draw[dashed] (1.5,4.2) -- (1.5, 0);
	\draw (1.5,0.13) -- (1.5,-0.13) node[anchor=north, text width=3cm, align=center] {closing time of day $t-1$};
	\draw[dashed] (4.5, 2) -- (4.5, 0);
	\draw (4.5,0.13) -- (4.5,-0.13) node[anchor=north, text width=3cm, align=center] {opening time of day $t$};
	\draw[dashed] (7.5, 4.2) -- (7.5, 0);
	\draw (7.5,0.13) -- (7.5,-0.13) node[anchor=north, text width=3cm, align=center] {time on the highest price};
	\draw[dashed] (10.5, 0.985) -- (10.5, 0);
	\draw (10.5,0.13) -- (10.5,-0.13) node[anchor=north, text width=3cm, align=center] {time on the lowest price};
	\draw[dashed] (13,2.5) -- (13, 0);
	\draw (13,0.13) -- (13,-0.13) node[anchor=north, text width=3cm, align=center] {closing time of day $t$};
	
	%arrow
	\draw[<->, >=latex,red] (4.5, 2) -- (1.5, 2) node[midway, above,red] {\Large $r$};
	\draw[<->, >=latex, red] (7.5, 4.2) -- (1.5, 4.2) node[midway, above,red] {\Large $u$};
	\draw[<->, >=latex, red] (10.5, 0.985) -- (1.5, 0.985) node[midway, above,red] {\Large $v$};
	
	%point
	\draw[fill] (1.5,1.2) circle [radius=0.05] node[anchor=west] {$S_{t-1}=S^c_{t-1}$}; % closing price of day t-1
	\draw[fill] (4.5, 2) circle [radius=0.05] node[anchor=north west] {$S_{t-1+r}=S^o_t$}; % opening price of day t
	\draw[fill] (7.5, 4.2) circle [radius=0.05] node[anchor=south] {$S_{t-1+u}=S^h_t$}; % highest price
	\draw[fill] (10.5, 0.985) circle [radius=0.05] node[anchor=north west] {$S_{t-1+v}=S^l_t$}; % lowest price
	\draw[fill] (13,2.5) circle [radius=0.05] node[anchor=west] {$S_t=S^c_t$}; % closing price of day t
	
	\end{tikzpicture}
	\caption{Opening, the highest, the lowest, and closing prices on a continuous time axis. %Note that the date index $t$ is an integer.
	}
	\label{fig:price_graph}
\end{figure*}

\subsection{Return-Volatility Normalization}
\label{rvn}
%We must reduce the distributional discrepancy among input historical stock prices to address the distribution shift problem.
We propose \RVN (Return-Volatility Normalization) module to reduce the distributional discrepancy, inspired by GBM~\cite{bjork2019arbitrage}.
\RVN normalizes historical stock prices $S_{T-w+1:T}^o$, $S_{T-w+1:T}^h$, $S_{T-w+1:T}^l$, and $S_{T-w+1:T}^c$ for $w$ days up to day $T$ into normalized inputs $\epsilon_{T-w+1:T}^o$, $\epsilon_{T-w+1:T}^h$, $\epsilon_{T-w+1:T}^l$, and $\epsilon_{T-w+1:T}^c$, which serve as inputs to the backbone model.
Note that $\epsilon_{T-w+1:T}^o$, $\epsilon_{T-w+1:T}^h$, $\epsilon_{T-w+1:T}^l$, and $\epsilon_{T-w+1:T}^c$ are not random variables as in Equation~\eqref{discrete_gbm}, but rather samples in this section.
The backbone model is unaffected by distribution shift since the error terms share an identical distribution across samples.

For example, consider a stock observed over two different time periods.
In the first period, the prices generally range between 100 and 120, and the daily returns have a mean of 0.02 and a standard deviation of 0.05.
In the second period, the prices decrease and range between 60 and 80, and the daily returns have a mean of -0.01 and a standard deviation of 0.12.
The raw daily returns in the first period, measured on a higher price scale (around 110), are tightly concentrated around 0.02 with low variability, while in the second period, the returns observed on a lower price scale (around 70) are more dispersed around -0.01 with higher variability.
If we train the model on the first period and test it on the second period without any normalization, the model encounters a significant distribution shift in the level and dispersion of daily returns, as well as in the its price scale.
After applying return-volatility normalization, the daily returns in both periods are transformed to have an identical distribution.
This process enables the backbone network to learn predictive patterns that remain stable even when the distribution of the same stock changes over time.

To normalize the stock prices with the sample characteristics, we must first define the sample characteristics.
In the widely used stock modeling method GBM, stock characteristics are defined as return, volatility, and price scale.
We assume that stock prices follow a GBM, except for the independence assumption (refer to property \ref{p2} of GBM in Section~\ref{sec:related:gbm}).
We relax this assumption since temporal dependencies exist due to market microstructure effects and autocorrelation in prices.
Observing Equation \eqref{discrete_gbm}, the equation consists of the variables $\mu$ (return), $\sigma$ (volatility), $S_{t-1}$ (price scale), and $\epsilon_{t}$ (error term).
Note that 1) $\mu$, $\sigma$, and $S_{t-1}$ represent sample characteristics, and 2) $\epsilon_{t}$ is sampled from a standard normal distribution.
Therefore, we eliminate distributional discrepancy by constructing an input with error terms.

We first briefly explain the process of estimating sample characteristics in Section~\ref{sec:est} and then describe normalization methods for closing, opening, and the highest/lowest prices in Sections~\ref{sec:close}, \ref{sec:open}, and \ref{sec:high}, respectively.

\subsubsection{Return and Volatility Estimation.} \label{sec:est}
We first estimate the sample characteristics—return and volatility—based on GBM, using proposed \RVE (Return-Volatility Estimator), which is an attention-based module designed to accurately estimate return and volatility (see Section~\ref{rve} for details).
The estimation is expressed as follows.

\begin{equation*}
\widehat{\mu-\frac{1}{2}\sigma^2}, \hat{\sigma} = RVE_\mathbb{\phi}\left(S_{T-w+1:T}^o,S_{T-w+1:T}^h,S_{T-w+1:T}^l,S_{T-w+1:T}^c\right) \\
\end{equation*}

\noindent where $\mathbb{\phi}$ denotes the parameters of the \RVE module.

\subsubsection{Normalizing Closing Price.} \label{sec:close}
We obtain the error term $\epsilon_t^c$ by normalizing the logarithmic daily return using the estimated return and volatility, as shown in the following Lemma.
% $e_t^o \overset{\text{sample}}{\sim} \mathcal{N}(0, 1)$
\begin{lemma} \label{lemma:error_close}
	An error term $\epsilon_t^c$ which corresponds to the closing price at day $t$, is obtained as follows:
	\begin{equation} \label{eq:close_norm}
	\epsilon_t^c=\frac{1}{\hat{\sigma}}\left(\log\frac{S_{t}^c}{S_{t-1}^c}-\left(\widehat{\mu-\frac{1}{2}\sigma^2}\right)\right).
	\end{equation}
	Furthermore, \( \epsilon_t^c \) is a sample drawn from a standard normal distribution: $	\epsilon_t^c \overset{\text{sample}}{\sim} \mathcal{N}(0,1)$.
\end{lemma}

\begin{proof}
	By Equation~\eqref{discrete_gbm}, the closing price $S_{t}$ on day $t$ is given as follows:
	\begin{equation} \label{gbm_closing}
	S_{t}^c = S_{t-1}^c e^{\left( \mu - \frac{1}{2} \sigma^2 \right) + \sigma \epsilon_{t}^c}.
	\end{equation}
	\noindent Note that $\epsilon_{t}^c \overset{\text{sample}}{\sim} \mathcal{N}(0,1)$ by Lemma~\ref{lemma:discrete_gbm}.
	We rearrange Equation~\eqref{gbm_closing} as follows:
	\begin{equation*}
	\epsilon_t^c=\frac{1}{\sigma}\left(\log\frac{S_{t}^c}{S_{t-1}^c}-\left(\mu-\frac{1}{2}\sigma^2\right)\right).
	\end{equation*}
	\noindent We get the result by replacing $\mu-\frac{1}{2}\sigma^2$ and $\sigma$ with their estimations.
\end{proof}

The error term $\epsilon_t^c$ represents the error of the closing price on day $t$ relative to the closing price on day $t-1$.
Thus, we exploit $\epsilon_t^c$ as a feature for training in place of $S_t^c$, the closing price on day $t$.

\subsubsection{Normalizing Opening Price.} \label{sec:open}

Capturing the relationship between the previous day’s closing price and today’s opening price is essential since overnight information—e.g., disclosures, news, and foreign market movements after the close—is immediately reflected in today’s opening price.
A simple comparison of the opening–opening relationship fails to account for this overnight drift.
Therefore, modeling the closing–opening relationship first is critical to improve daily prediction performance.

A simple way to normalize the opening price is to apply Equation~\eqref{eq:close_norm} in the same way, as follows:

\begin{equation*}
\epsilon_t^o=\frac{1}{\hat{\sigma}}\left(\log\frac{S_{t}^o}{S_{t-1}^o}-\left(\widehat{\mu-\frac{1}{2}\sigma^2}\right)\right).
\end{equation*}
\noindent However, in this case, the error term $\epsilon_t^o$ represents the change from the opening price on day $t-1$ to the opening price on day $t$.
Then, the relationship between the opening and closing prices cannot be captured.
Therefore, we propose to map the opening and closing time onto a continuous time axis $[0,\infty)$, and associate each time with the corresponding price.
Then, we calculate error terms that capture the relationship between them.

Let $S_t$ represent the stock price at time $t$ on a continuous time axis $[0,\infty)$.
In this axis, $t\in\mathbb{N}$ represents the closing time of day $t$ where we regard $t$ as an integer.
That is, $S_t=S_t^c$ for $t\in\mathbb{N}$.
We associate the opening, highest, and lowest prices on day $t$ with their respective time points on this time axis.
Let $0\leq r,u,v\leq1$ represent the time intervals between the closing price on day $t-1$ and the opening, highest, and lowest prices on day $t$, respectively, as shown in Figure~\ref{fig:price_graph}.
Then, on the continuous time axis, $S_{t-1+r}=S_t^o$, $S_{t-1+u}=S_t^h$, and $S_{t-1+v}=S_t^l$.

By Equation~\eqref{gbm}, the opening price $S_{t-1+r}$ is expressed in terms of the previous day's closing price $S_{t-1}$:

\begin{equation} \label{open_gbm}
S_{t-1+r} = S_{t-1} e^{\left( \mu - \frac{1}{2} \sigma^2 \right)r + \sigma \epsilon_t^{o}}
\end{equation}

\noindent where $\epsilon_t^{o}\overset{\text{sample}}{\sim}\mathcal{N}(0, r)$ is the error term sampled from $B_{t-1+r} - B_{t-1}$.
$\epsilon_t^{o}$ is determined if $r$ is known. % since other variables ($S_{t-1+r}$, $S_{t-1}$, $\mu$, and $\sigma$) are already known.
However, determining the time interval $r$ between the opening and previous day's closing prices is not straightforward.
We might consider setting $r$ to 0 since the opening price is the immediate price following the previous day's closing price.
However, the opening price differs from the previous day's closing price due to price fluctuations during after-hours trading, meaning that $r$ cannot be set to 0.
Therefore, we need to estimate $r$ before normalizing the opening price.
We estimate $r$ based on the assumption that stock prices follow a geometric modeling with respect to time.
%Thus, we estimate $r$ with the following Lemma.

\begin{lemma} \label{lemma:estimate_r}
	By minimizing a mean squared error, $r$ is estimated as follows:
	
	\begin{equation} \label{eq:est_r}
	\hat{r}=\frac{\mathbb{E}\left[\log{\frac{S_{t}^c}{S_{t-1}^c}}\log{\frac{S_{t}^o}{S_{t-1}^c}}\right]}{\mathbb{E}\left[\left(\log{\frac{S_{t}^c}{S_{t-1}^c}}\right)^2\right]}.
	\end{equation}
	
\end{lemma}

\begin{proof}
	According to the geometric modeling of stock prices, we express the opening price on day $t$ as follows:
	
	\begin{equation} \label{eq:open_close}
	S_{t-1+r}=S_{t-1}\left(\frac{S_{t}}{S_{t-1}}\right)^{r}.
	\end{equation}

	\noindent This expression is based on a simplifying assumption that the stock price follows a deterministic path: $S_t=S_0 e^{\kappa t}$ for some $\kappa\in\mathbb{R}$.
	While this assumption does not fully align with the GBM and is slightly inconsistent with Equations~\eqref{discrete_gbm} and \eqref{open_gbm},
	it offers a tractable approximation that allows for a closed-form estimation of $r$, which is useful in practice.
	
	\noindent By taking the log of Equation~\eqref{eq:open_close}, we rewrite it as follows:
	
	\begin{equation*}
	\log{\frac{S_{t-1+r}}{S_{t-1}}}=r\log{\frac{S_{t}}{S_{t-1}}}.
	\end{equation*}
	\noindent To find the best $r$ that satisfies the equation above, we minimize the following mean squared error:
	
	\begin{equation} \label{eq:mse}
	\mathbb{E}\left[\left(r\log{\frac{S_{t}}{S_{t-1}}}-\log{\frac{S_{t-1+r}}{S_{t-1}}}\right)^2\right].
	\end{equation}
	
	\noindent Equation~\eqref{eq:mse} attains its minimum at the point where the derivative with respect to $r$ is equal to zero since it is convex with respect to $r$:
	
	\begin{eqnarray}
	\quad \quad \quad 0&=&\frac{d}{dr}\mathbb{E}\left[\left(r\log{\frac{S_{t}}{S_{t-1}}}-\log{\frac{S_{t-1+r}}{S_{t-1}}}\right)^2\right] \nonumber \\
	&=&\mathbb{E}\left[2\log{\frac{S_{t}}{S_{t-1}}}\left(r\log{\frac{S_{t}}{S_{t-1}}}-\log{\frac{S_{t-1+r}}{S_{t-1}}}\right)\right] \nonumber \\
	\label{eq:diff} &=&\mathbb{E}\left[2\log{\frac{S_{t}^c}{S_{t-1}^c}}\left(r\log{\frac{S_{t}^c}{S_{t-1}^c}}-\log{\frac{S_{t}^o}{S_{t-1}^c}}\right)\right]
	\end{eqnarray}
	
	\noindent Solving Equation \eqref{eq:diff}, we obtain the result.
\end{proof}

As a result, we obtain the error term $\epsilon_t^o$ with the following Lemma.

\begin{lemma} \label{lemma:error_open}
	Assume that Equation~\eqref{eq:est_r} accurately estimates $r$, i.e., $\hat{r}=r$.
	Then, the error term $\epsilon_t^o$  corresponding to the opening price on day $t$, is obtained as follows:
	\begin{equation} \label{eq:open_norm}
	\epsilon_t^{o}=\frac{1}{\hat{\sigma}}\left(\log\frac{S_{t}^o}{S_{t-1}^c}-\left(\widehat{\mu-\frac{1}{2}\sigma^2}\right)\hat{r}\right).
	\end{equation}
	Furthermore, \( \frac{1}{\sqrt{\hat{r}}}\epsilon_t^o \) is a sample drawn from a standard normal distribution:
	\begin{equation*}
	\frac{1}{\sqrt{\hat{r}}}\epsilon_t^o \overset{\text{sample}}{\sim} \mathcal{N}(0,1).
	\end{equation*}
\end{lemma}

\begin{proof}
	Rewriting Equation~\eqref{open_gbm}, we obtain the following:
	\begin{equation*}
	\epsilon_t^o=\frac{1}{\sigma}\left(\log\frac{S_{t-1+r}}{S_{t-1}}-\left(\mu-\frac{1}{2}\sigma^2\right)r\right).
	\end{equation*}
	\noindent Since $\epsilon_t^{o}\overset{\text{sample}}{\sim}\mathcal{N}(0, r)$, $\frac{1}{\sqrt{\hat{r}}}\epsilon_t^o \overset{\text{sample}}{\sim} \mathcal{N}(0,1)$.
	We obtain the result by replacing $\mu-\frac{1}{2}\sigma^2$, $\sigma$, and $r$ with their estimations.
\end{proof}

\subsubsection{Normalizing the Highest and Lowest Prices.}\label{sec:high}
The highest price also cannot be normalized to the error term using the previous day's highest price, since it fails to capture the relationship with the closing price.
Furthermore, it is also impossible to normalize it in the same way as the opening price, as shown in Equation~\eqref{eq:open_norm}, since the time interval $u$ between the highest and previous day's closing prices is unknown.

To address this, we define $\epsilon_t^h$ and normalize the highest price as follows:

\begin{equation*}
\epsilon_t^h = \frac{1}{\hat{\sigma}} \log \frac{S_t^h}{S_{t-1}^c}.
\end{equation*}

\noindent This definition is justified as follows.
Since $S_t^h$ is the highest price on day $t$,

\begin{equation} \label{eq:high_gbm}
S_t^h = \max_{0<u\leq 1} S_{t-1}^c e^{\left(\mu - \frac{1}{2} \sigma^2\right)u + \sigma B_u}.
\end{equation}

\noindent We simplify Equation~\eqref{eq:high_gbm} by ignoring the drift term $\left(\mu - \frac{1}{2} \sigma^2\right)u$, as its contribution is negligible compared to the volatility term $\sigma B_u$ on a daily scale~\cite{andersen1998answering, andersen2001distribution}.
These small values indicate that the drift's effect is minor in intraday settings.
To empirically justify this simplification, we approximate the distribution of the ratio $\left|\frac{\left(\mu - \frac{1}{2} \sigma^2\right)u}{\sigma B_u}\right|$. The results show that this ratio is less than $\frac{1}{5}$ in 74.1\% of the cases in the U.S., 78.7\% in China, 75.3\% in the U.K., and 73.0\% in Korea.
Thus, we get

\begin{equation*}
S_t^h \approx \max_{0<u\leq 1} S_{t-1}^c e^{\sigma B_u} \quad \Rightarrow \quad \frac{1}{\sigma} \log \frac{S_t^h}{S_{t-1}^c} \approx \max_{0<u\leq 1} B_u.
\end{equation*}

\noindent Since $\max_{0<u\leq 1} B_u$ follows a distribution independent of specific samples, the resulting $\epsilon_t^h$ also follows a consistent distribution across different samples.
%Therefore, $\epsilon_t^h$ is normalized in a way that ensures its distribution is stable across data instances.

We also normalize the lowest price in the same way:

\begin{equation*}
\epsilon_t^l = \frac{1}{\hat{\sigma}} \log \frac{S_t^l}{S_{t-1}^c}.
\end{equation*}

\noindent Note that $\epsilon_t^l$ also follows a consistent distribution across different samples.

\subsection{Return-Volatility Estimator}
\label{rve}
\RVE (Return-Volatility Estimator) is designed to accurately estimate sample characteristics from historical stock prices.
Specifically, \RVE takes historical prices as input and estimates the distribution parameters: return and volatility.
Stock prices sometimes deviate temporarily from their true distribution due to external shocks.
For example, stock prices experience a temporary surge due to expectations around a new product launch, or a sharp decline due to institutional investors liquidating their positions for short-term gains.
These temporary fluctuations create noise in the data, making it difficult to accurately estimate the underlying distribution of stock prices.
We propose to estimate the distribution parameters exploiting a weighted average based on attention weights to handle such anomalies flexibly.

We observe that $\log\frac{S_{t}^c}{S_{t-1}^c}=\mu-\frac{1}{2}\sigma^2+\sigma\epsilon_t^c \sim \mathcal{N}(\mu-\frac{1}{2}\sigma^2,\sigma^2)$ by applying Equation~\eqref{discrete_gbm} to closing prices.
According to this observation, a simple approach to estimate $\mu-\frac{1}{2}\sigma^2$ and $\sigma$ is to compute the arithmetic mean as follows:

\begin{eqnarray}
\widehat{\mu-\frac{1}{2}\sigma^2}&=&\frac{1}{w}\sum_{t=T-w+1}^{T}\log\frac{S_{t}^c}{S_{t-1}^c} \label{eq:simple_return} \\
\hat{\sigma}&=&\sqrt{\frac{1}{w}\sum_{i=T-w+1}^{T}\left(\log\frac{S_{t}^c}{S_{t-1}^c}-\widehat{\mu-\frac{1}{2}\sigma^2}\right)^2}. \label{eq:simple_volatility}
\end{eqnarray}

However, this approach is significantly affected by anomalies that deviate from the true distribution, making it difficult to accurately estimate the distribution parameters.
Our main idea is to estimate $\mu-\frac{1}{2}\sigma^2$ and $\sigma$ using a weighted average based on attention weights as follows:

\begin{eqnarray}
\widehat{\mu-\frac{1}{2}\sigma^2}&=&\sum_{t=T-w+1}^{T}\alpha_t\log\frac{S_{t}^c}{S_{t-1}^c} \label{eq:return} \nonumber \\
\hat{\sigma}&=&\sqrt{\sum_{t=T-w+1}^{T}\alpha_t\left(\log\frac{S_{t}^c}{S_{t-1}^c}-\widehat{\mu-\frac{1}{2}\sigma^2}\right)^2} \nonumber \label{eq:volatility}
\end{eqnarray}

\noindent where $\alpha_t$ is the attention weight at time step $t$, calculated by Equation~\eqref{eq:attn_weight}. Note that $\sum_{t=T-w+1}^{T}\alpha_t=1$. This approach flexibly handles anomalies by dynamically adjusting the average weights.

\textbf{Attention module structure.} We describe the structure of the attention module to compute the attention weight $\alpha_t$.
We exploit attention LSTM \cite{feng2019enhancing} to generate query and key vectors that effectively capture sequential dependencies in time-series data.

We begin by transforming price feature vector $\mathbf{x}_t \in \mathbb{R}^D$ at each time step using a fully connected layer with a hyperbolic tangent activation function,  $\tilde{\mathbf{x}}_t = tanh(\mathbf{W} \mathbf{x}_t + \mathbf{b})$,
where the parameters $\mathbf{W} \in \mathbb{R}^{E \times D}$ and $\mathbf{b} \in \mathbb{R}^E$ are shared across all time steps $t$.
This transformation produces a new feature representation before the data are passed into the LSTM, improving the network's ability to capture temporal patterns \cite{feng2019enhancing}.

The feature vectors $\tilde{\mathbf{x}}_{T-w+1},\ldots,\tilde{\mathbf{x}}_{T}$ are then passed into the LSTM and transformed into hidden state vectors $\mathbf{h}_{T-w+1},\ldots,\mathbf{h}_{T}$.
The hidden state vector $\mathbf{h}_{T}$ from the last time step represents past information and serves as the query vector.
Finally, the attention weight $\alpha_t$ at time step $t$ is calculated as follows:

\begin{equation}
\alpha_t = \frac{e^{\mathbf{h}_t^\top \mathbf{h}_{T}}}{\sum_{i=T-w+1}^{T} e^{\mathbf{h}_i^\top \mathbf{h}_{T}}}. \label{eq:attn_weight}
\end{equation}

\subsection{Backbone and Return-Volatility Denormalization}
\label{rvd}

The neural network backbone $f_\mathbf{\theta}$ receives the normalized error terms of historical stock prices as input, specifically $\epsilon_{T-w+1:T}^o$, $\epsilon_{T-w+1:T}^h$, $\epsilon_{T-w+1:T}^l$, and $\epsilon_{T-w+1:T}^c$.
These error terms are generated by the \RVN module and share an identical distribution across samples, effectively mitigating distribution shift.
The backbone model outputs the predicted future error term $\hat{\epsilon}_{T+1}^c$, which does not contain any sample-specific characteristics such as return or volatility.

\begin{equation*}
\hat{\epsilon}_{T+1}^c=f_{\mathbf{\theta}}(\epsilon_{T-w+1:T}^o, \epsilon_{T-w+1:T}^h, \epsilon_{T-w+1:T}^l, \epsilon_{T-w+1:T}^c)
\end{equation*}

The backbone can be instantiated with any time series prediction model, such as LSTM, GRU, or Transformer variants, allowing \method to leverage their strengths while controlling for distribution shift.

With the predicted future error term $\hat{\epsilon}_{T+1}^c$,
the \RVD (Return-Volatility Denormalization) module denormalizes it to produce the final prediction $\hat{S}_{T+1}^c$.
It reintegrates the information of sample characteristics lost in the normalization process into the prediction of the backbone model.
\RVD works by replacing the BM part of the GBM with the predicted future error term.

\begin{equation*}
\hat{S}_{T+1}^c = S_{T}^c e^{\left( \widehat{\mu - \frac{1}{2} \sigma^2} \right) + \hat{\sigma} \hat{\epsilon}_{T+1}^c}.
\end{equation*}

Due to the independence assumption that future prices are independent of past prices, GBM cannot predict future prices.
However, the model gains forecasting ability by eliminating the independence assumption and replacing the BM part with the predicted future error term.
\method considers both learned patterns of the historical error terms and inherent price dynamics reflected in the sample characteristics.

\subsection{Loss Function}
The final loss function is designed to optimize the parameters of the backbone model $\mathbf{\theta}$ and the parameters of the \RVE module $\mathbf{\phi}$.
The loss is given by

\begin{equation*} \label{eq:loss}
\mathbb{E}\left[\left(\frac{\hat{S}_{T+1}^c}{S_{T}^c} - \frac{S_{T+1}^c}{S_{T}^c}\right)^2
+ \beta \left(\frac{1}{w} \sum_{t=T-w+1}^{T} \log \frac{S_{t}^c}{S_{t-1}^c} - \widehat{\mu - \frac{1}{2} \sigma^2}\right)^2 \right].
\end{equation*}

The loss is composed of two parts.
The first part is the MSE loss for daily return prediction.
The second part is a guidance loss to give robustness to \RVE.
We observe that \method with returns calculated through arithmetic averages (Equation~\eqref{eq:simple_return}) also leads to the performance improvement.
Therefore, we propose to guide $\widehat{\mu-\frac{1}{2}\sigma^2}$ to resemble the arithmetic mean of the logarithmic daily return through minimizing guidance loss.
These joint objectives enable end-to-end training of the backbone and the \RVE module.
$\beta$ is a hyperparameter to control the effect of the guidance loss.
Algorithm~\ref{algorithm:training} illustrates the whole training process.

\begin{algorithm}[t]
	\caption{Training Algorithm for \method}
	\label{algorithm:training}
	
	\begin{flushleft}
		\textbf{Input}: historical stock price dataset \\
		\textbf{Output}: optimal backbone parameters $\mathbf{\theta}$ and RVE module parameters $\mathbf{\phi}$
	\end{flushleft}
	
	\begin{algorithmic}[1]
		\STATE \textbf{for} each mini-batch \textbf{do}
		\STATE \quad \textbf{for} each sample \textbf{do}
		\STATE \quad \quad $\widehat{\mu-\frac{1}{2}\sigma^2}, \hat{\sigma} \leftarrow RVE_\mathbb{\phi}\big(S_{T-w+1:T}^o,S_{T-w+1:T}^h,S_{T-w+1:T}^l,S_{T-w+1:T}^c\big)$
		\STATE \quad \quad $\epsilon_{T-w+1:T}^o, \epsilon_{T-w+1:T}^h, \epsilon_{T-w+1:T}^l, \epsilon_{T-w+1:T}^c \leftarrow$
		\STATE \quad \quad \quad $RVN\Big( S_{T-w+1:T}^o, S_{T-w+1:T}^h, S_{T-w+1:T}^l, S_{T-w+1:T}^c \, ; \, \widehat{\mu - \frac{1}{2}\sigma^2}, \hat{\sigma} \Big)$
		\STATE \quad \quad $\hat{\epsilon}_{T+1}^c \leftarrow f_{\mathbf{\theta}}(\epsilon_{T-w+1:T}^o,\epsilon_{T-w+1:T}^h,\epsilon_{T-w+1:T}^l,\epsilon_{T-w+1:T}^c)$
		\STATE \quad \quad $\hat{S}_{T+1}^c \leftarrow S_{T}^c e^{\left( \widehat{\mu - \frac{1}{2} \sigma^2} \right) + \hat{\sigma}\hat{\epsilon}_{T+1}^c}$
		\STATE \quad \textbf{end for}
		\STATE \quad Update ${\mathbf{\theta},\mathbf{\phi}}$ by minimizing $\mathbb{E}\big[\left(\frac{\hat{S}_{T+1}^c}{S_{T}^c} - \frac{S_{T+1}^c}{S_{T}^c}\right)^2$
		\STATE \quad \quad $+ \beta \left(\frac{1}{w} \sum_{t=T-w+1}^{T} \log \frac{S_{t}^c}{S_{t-1}^c} - \widehat{\mu - \frac{1}{2} \sigma^2}\right)^2 \big]$ using Adam~\cite{kingma2014adam}
		\STATE \textbf{end for}
	\end{algorithmic}	
\end{algorithm} 
\section{Experiments} \label{experiments}

\begin{table}
	\centering
	\caption{
		Summary of datasets.
	}
	\begin{threeparttable}
		\begin{tabular}{l|r|r|cccc}
			\toprule
			\textbf{Data} &
			\textbf{Stocks} &
			\textbf{Days} &
			\textbf{Dates} \\
			\midrule
			U.S.\textsuperscript{1} &
			178 &
			5035 &
			2003-01-02 to 2022-12-30 \\
			China\textsuperscript{1} &
			171 &
			2430 &
			2011-01-06 to 2020-12-31 \\
			U.K.\textsuperscript{1} &
			21 &
			2383 &
			2014-01-06 to 2024-01-10 \\
			Korea\textsuperscript{1} &
			220 &
			2170 &
			2014-01-02 to 2022-11-01 \\
			\bottomrule
		\end{tabular}
		\begin{tablenotes} \footnotesize
			\item[1]\url{https://github.com/AnonymousCoder2357/ReVol}
		\end{tablenotes}
	\end{threeparttable}
	\label{data}
\end{table}

\begin{table}
	\centering
	\caption{
		An input feature vector $\mathbf{x}_t$ at day $t$.
	}
	\begin{threeparttable}
		\begin{tabular}{l|l}
			\toprule
			\textbf{Feature} &
			\textbf{Calculation} \\
			\midrule
			$x_t^{o}$ &
			$x_t^{o}=S_t^o/S_t^c-1$ \\
			$x_t^{h}$ &
			$x_t^{h}=S_t^h/S_t^c-1$ \\
			$x_t^{l}$ &
			$x_t^{l}=S_t^l/S_t^c-1$ \\
			$x_t^{c}$ &
			$x_t^{c}=S_t^c/S_{t-1}^c-1$ \\
			\bottomrule
		\end{tabular}
	\end{threeparttable}
	\label{normalization}
\end{table}

\begin{table*}[!t]
	\centering
	\setlength{\tabcolsep}{4pt}
	\caption{
		Stock prediction performance of \method and baselines measured by Information Coefficient (IC), Rank Information Coefficient (RIC), and Sharpe Ratio (SR). \method consistently improves the performance of backbone models in most cases, showing an average improvement of over 0.03 in IC and over 0.7 in SR.
	}
	\begin{tabular}{l|rrr|rrr|rrr|rrr}
		\toprule
		\multirow{2}{*}{Method}
		& \multicolumn{3}{c}{U.S.}
		& \multicolumn{3}{c}{China}
		& \multicolumn{3}{c}{U.K.}
		& \multicolumn{3}{c}{Korea} \\
		\cmidrule(lr){2-4} \cmidrule(lr){5-7} \cmidrule(lr){8-10} \cmidrule(lr){11-13}
		& IC & RIC & SR
		& IC & RIC & SR
		& IC & RIC & SR
		& IC & RIC & SR \\
		\midrule
		LSTM \cite{nelson2017stock}
		& 0.019 & \textbf{0.026} & 0.642
		& 0.015 & 0.004 & 1.283
		& 0.040 & 0.014 & 0.820
		& 0.043 & 0.029 & 2.092 \\
		\textbf{+ \method}
		& \textbf{0.022} & 0.021 & \textbf{1.001}
		& \textbf{0.042} & \textbf{0.030} & \textbf{2.628}
		& \textbf{0.108} & \textbf{0.056} & \textbf{1.035}
		& \textbf{0.044} & \textbf{0.045} & \textbf{2.338} \\
		\midrule
		GRU \cite{cho2014learning}
		& 0.020 & \textbf{0.027} & 0.656
		& 0.015 & 0.008 & 1.181
		& 0.037 & 0.020 & 0.382
		& 0.043 & 0.023 & 2.008 \\
		\textbf{+ ReVol}
		& \textbf{0.022} & 0.022 & \textbf{0.901}
		& \textbf{0.032} & \textbf{0.024} & \textbf{2.090}
		& \textbf{0.092} & \textbf{0.058} & \textbf{1.289}
		& \textbf{0.044} & \textbf{0.046} & \textbf{2.348} \\
		\midrule
		ALSTM \cite{feng2019enhancing}
		& 0.017 & \textbf{0.022} & 0.796
		& 0.022 & 0.014 & 1.281
		& 0.062 & 0.028 & 1.006
		& 0.041 & 0.028 & 1.961 \\
		\textbf{+ \method}
		& \textbf{0.022} & 0.021 & \textbf{0.943}
		& \textbf{0.039} & \textbf{0.026} & \textbf{2.436}
		& \textbf{0.112} & \textbf{0.059} & \textbf{1.358}
		& \textbf{0.047} & \textbf{0.049} & \textbf{2.728} \\
		\midrule
		Vanilla Transformer \cite{vaswani2017attention}
		& 0.015 & \textbf{0.022} & 0.768
		& 0.011 & 0.005 & 1.110
		& 0.046 & 0.026 & 0.804
		& 0.033 & 0.014 & 1.306 \\
		\textbf{+ ReVol}
		& \textbf{0.020} & 0.019 & \textbf{0.865}
		& \textbf{0.032} & \textbf{0.017} & \textbf{2.608}
		& \textbf{0.127} & \textbf{0.062} & \textbf{1.594}
		& \textbf{0.039} & \textbf{0.046} & \textbf{1.315} \\
		\midrule
		DTML \cite{yoo2021accurate}
		& 0.016 & \textbf{0.017} & 0.834
		& 0.013 & 0.010 & 1.109
		& -0.004 & -0.004 & 0.115
		& 0.006 & -0.042 & 0.106 \\
		\textbf{+ \method}
		& \textbf{0.018} & 0.016 & \textbf{0.886}
		& \textbf{0.048} & \textbf{0.043} & \textbf{2.454}
		& \textbf{0.095} & \textbf{0.046} & \textbf{0.879}
		& \textbf{0.035} & \textbf{0.044} & \textbf{1.674} \\
		\midrule
		MASTER \cite{li2024master}
		& 0.013 & 0.014 & 0.753
		& 0.002 & 0.014 & 0.542
		& 0.019 & 0.001 & 0.345
		& 0.005 & 0.015 & 0.359 \\
		\textbf{+ \method}
		& \textbf{0.015} & \textbf{0.015} & \textbf{0.794}
		& \textbf{0.035} & \textbf{0.019} & \textbf{2.248}
		& \textbf{0.102} & \textbf{0.050} & \textbf{1.213}
		& \textbf{0.038} & \textbf{0.042} & \textbf{1.688} \\
		\bottomrule
	\end{tabular}
	\label{table:perf}
\end{table*}

\begin{table*}[!t]
	\centering
	\setlength{\tabcolsep}{4pt}
	\caption{
		Performance comparison on Information Coefficient (IC), Rank Information Coefficient (RIC), and Sharpe Ratio (SR) with the state-of-the-art normalization methods. \method outperforms other methods, demonstrating better capability to reduce distributional discrepancy in stock price data.
	}
	\begin{tabular}{l|rrr|rrr|rrr|rrr}
		\toprule
		\multirow{2}{*}{Method}
		& \multicolumn{3}{c}{U.S.}
		& \multicolumn{3}{c}{China}
		& \multicolumn{3}{c}{U.K.}
		& \multicolumn{3}{c}{Korea} \\
		\cmidrule(lr){2-4} \cmidrule(lr){5-7} \cmidrule(lr){8-10} \cmidrule(lr){11-13}
		& IC & RIC & SR
		& IC & RIC & SR
		& IC & RIC & SR
		& IC & RIC & SR \\
		\midrule
		RevIN~\cite{kim2021reversible}
		& 0.010 & 0.011 & 0.749
		& 0.030 & 0.024 & 2.096
		& 0.013 & 0.015 & 0.283
		& 0.023 & 0.021 & 1.511 \\
		Dish-TS~\cite{fan2023dish}
		& 0.012 & \textbf{0.020} & 0.518
		& 0.020 & 0.032 & 1.477
		& 0.015 & 0.019 & 0.250
		& 0.015 & 0.002 & 1.296 \\
		\midrule
		\textbf{\method} (proposed)
		& \textbf{0.018} & 0.016 & \textbf{0.886}
		& \textbf{0.048} & \textbf{0.043} & \textbf{2.454}
		& \textbf{0.095} & \textbf{0.046} & \textbf{0.879}
		& \textbf{0.035} & \textbf{0.044} & \textbf{1.674} \\
		\bottomrule
	\end{tabular}
	\label{table:dishts}
\end{table*}

We present experimental results to answer the following research questions about \method:

\begin{itemize*}
	\item[Q1.] \textbf{Performance (Section~\ref{perfromance}).}
		Does \method outperform previous methods in stock price prediction?
	\item[Q2.] \textbf{Comparison with other normalization methods (Section~\ref{sec:dishts}).}
		How well does \method address distribution shifts in stock price data compared to other normalization methods?
	\item[Q3.] \textbf{Hyperparameter sensitivity (Section~\ref{hyperparameter}).}
	How sensitive is \method to hyperparameters?
	\item[Q4.] \textbf{Noise-aware behavior of attention weights (Section~\ref{sec:attention_noise}).} Does the attention mechanism in the RVE module effectively suppress noisy inputs during return-volatility estimation?
	
	\item[Q5.] \textbf{Ablation study (Section~\ref{ablation}).}
	Does each module of \method contribute to improving the performance for stock price prediction?	
\end{itemize*}

\subsection{Experimental Setup}
%We describe datasets, competitors, and evaluation metrics for experiments.

\textbf{Datasets.} We use stock market datasets collected from the stock markets of the United States, China, United Kingdom, and Korea. The summary of the datasets is presented in Table \ref{data}.
The goal is to predict the daily return $\frac{S_{T+1}^c}{S_{T}^c}$ given the price features up to day $T$. We split the datasets in chronological order into 70\% for training, 10\% for validation, and 20\% for testing.

\textbf{Data normalization.} We generate an input vector $\mathbf{x_t}$ at day $t$ for baselines using widely used price ratio-based normalization method \cite{feng2019enhancing, yoo2021accurate}.
In contrast, \method exploits the normalization technique described in Section \ref{rvn}.
The normalization method for competitors is shown in Table \ref{normalization}.
$x_t^{o}$, $x_t^{h}$, and $x_t^{l}$ indicate the relative values of the opening, highest, and lowest prices in comparison to closing price at day $t$, respectively. $x_t^{c}$ indicates the relative value of the closing price at day $t$ in comparison to closing price at day $t-1$.

\textbf{Baselines.} \method is a model-agnostic method applied to neural network backbones for stock price prediction. We apply \method to the following baselines to validate its performance.

\begin{itemize*}
	\item \textbf{LSTM} \cite{nelson2017stock} is a recurrent neural network designed to capture long-term dependencies in time series data.
	
	\item \textbf{GRU} \cite{cho2014learning} is a simplified variant of LSTM that uses gating mechanisms to effectively capture sequential dependencies in time series data with fewer parameters.
	
	\item \textbf{ALSTM} \cite{feng2019enhancing} enhances the LSTM model by incorporating an attention mechanism, allowing it to focus on the most relevant time steps for improved stock price prediction.
	
	\item \textbf{Vanilla Transformer} \cite{vaswani2017attention} is a self-attention-based architecture originally proposed for machine translation, which can model long-range dependencies in sequences without relying on recurrence.
	
	\item \textbf{DTML} \cite{yoo2021accurate} leverages a transformer structure to capture temporal and inter-stock correlations with global market context for accurate prediction.
	
	\item \textbf{MASTER} \cite{li2024master} is a transformer-based approach for stock price forecasting that effectively captures momentary and cross-time stock correlations through two-step self-attentions in temporal and stock dimensions, respectively.
\end{itemize*}

\textbf{Evaluation metrics.} We conduct 10 experiments with different random seeds ranging from 0 to 9, and report the average across all cases. We evaluate the results of stock price prediction using Information Coefficient (IC) and Rank Information Coefficient (RIC).
IC and RIC are the dominant metrics in finance, which measure the linear relationships between the predicted values and the actual values with Pearson correlation coefficient and Spearman correlation coefficient, respectively.

We evaluate the portfolio's performance using annualized Sharpe Ratio (SR). The portfolio is rebalanced at the close of each day to allocate a weight of 0.2 to the top 5 stocks with the highest predicted returns.

\textbf{Hyperparameters. } We optimize the network parameters using the Adam~\cite{kingma2014adam} optimizer and apply early stopping based on the IC of the validation set.
We conduct a hyperparameter search on learning rate, weight decay, the number of attention heads (for transformer-based models), window size, hidden layer size, and guidance loss weight in $\{0.00001, 0.0001, 0.001\}$, $\{0, 0.0001, 0.001, 0.01\}$, $\{1, 2, 4, 8\}$, $\{8, 16, 24, 32, 40\}$, $\{128, 256, 512\}$, and $\{0, 0.25, 0.5, 1\}$, respectively.

\subsection{Performance (Q1)} \label{perfromance}
Table~\ref{table:perf} presents the performance comparison of \method and the baseline models for stock price prediction across four markets (U.S., China, U.K., and Korea). \method consistently improves performance across all backbone models and metrics, demonstrating its effectiveness in addressing distribution shifts.
On average, \method improves IC by over 0.03 and SR by more than 0.7 compared to the baselines, confirming its ability to enhance predictive accuracy and profitability.

Notably, \method improves model stability and generalization across diverse market conditions, achieving consistent gains across multiple backbones.
This robustness to distributional variations makes \method a practical and scalable solution for real-world financial forecasting tasks.

\subsection{Comparison with Other Normalization Methods (Q2)} \label{sec:dishts}
Table~\ref{table:dishts} compares the performance of \method with two state-of-the-art normalization methods, RevIN~\cite{kim2021reversible} and Dish-TS~\cite{fan2023dish}. RevIN, Dish-TS and \method utilize DTML as the backbone model.

\method shows consistent improvement in IC and SR across all markets, outperforming the other normalization methods.
As illustrated in Figure~\ref{figure:dist}, \method accurately aligns the distribution shapes between the training and test sets while the other method matches only distribution parameters. This precise alignment reduces distributional discrepancies that degrade predictive accuracy.
Additionally, unlike other methods which normalize features independently and lose inter-feature correlations, \method preserves these relationships.
This enables more accurate and stable predictions.

\begin{figure}
	\centering
	\begin{minipage}{0.37\linewidth}
		\centering
		\includegraphics[width=0.99\linewidth, trim=0 0 1mm 1mm, clip]{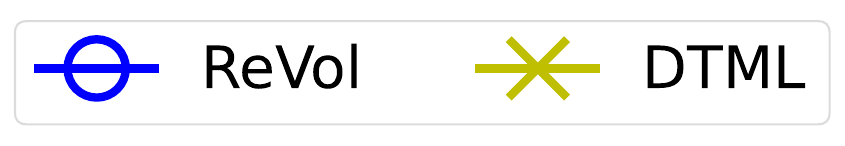}
	\end{minipage}
	\begin{minipage}{0.98\linewidth}
		\centering
		\begin{subfigure}{0.49\textwidth}
			\includegraphics[width=\textwidth]{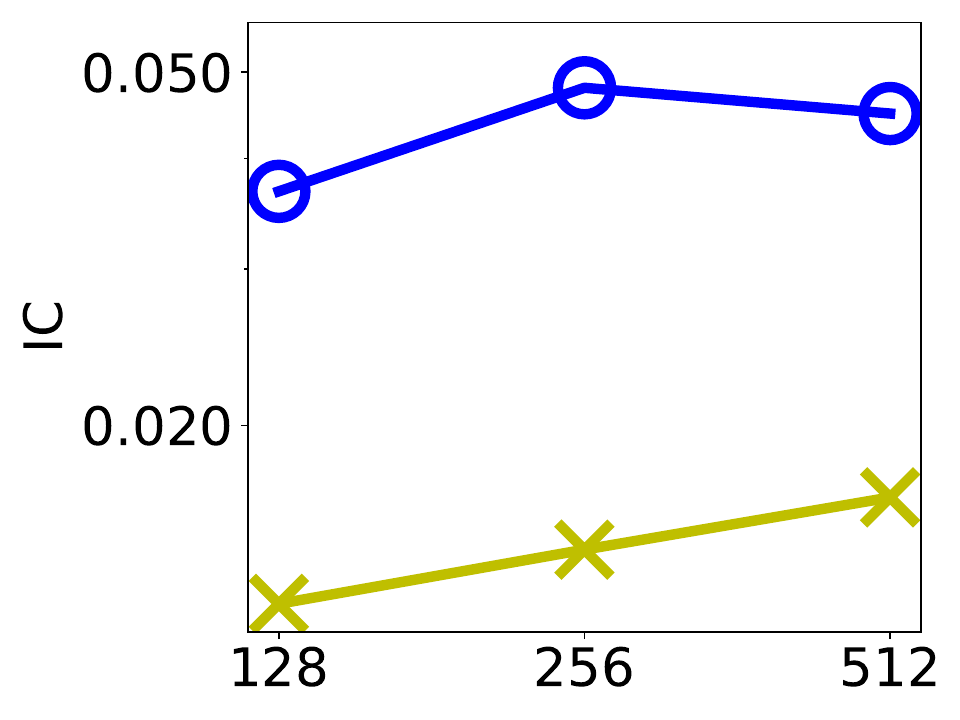}
			\caption{Hidden layer size}
		\end{subfigure} \hfill
		\begin{subfigure}{0.49\textwidth}
			\includegraphics[width=\textwidth]{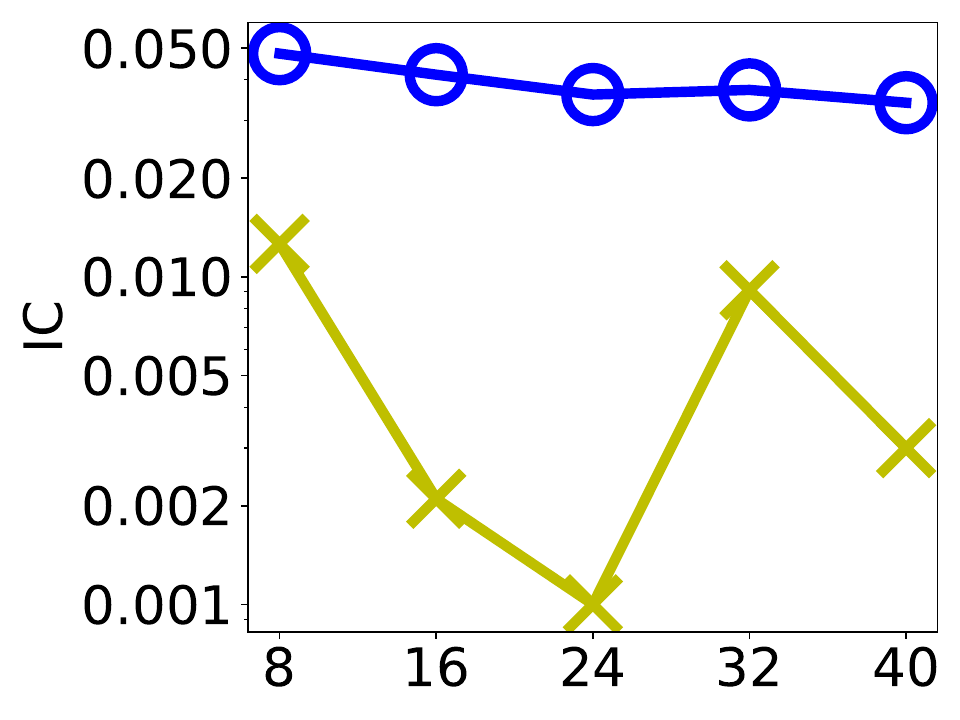}
			\caption{Window size}
		\end{subfigure} \hfill
		\begin{subfigure}{0.49\textwidth}
			\includegraphics[width=\textwidth]{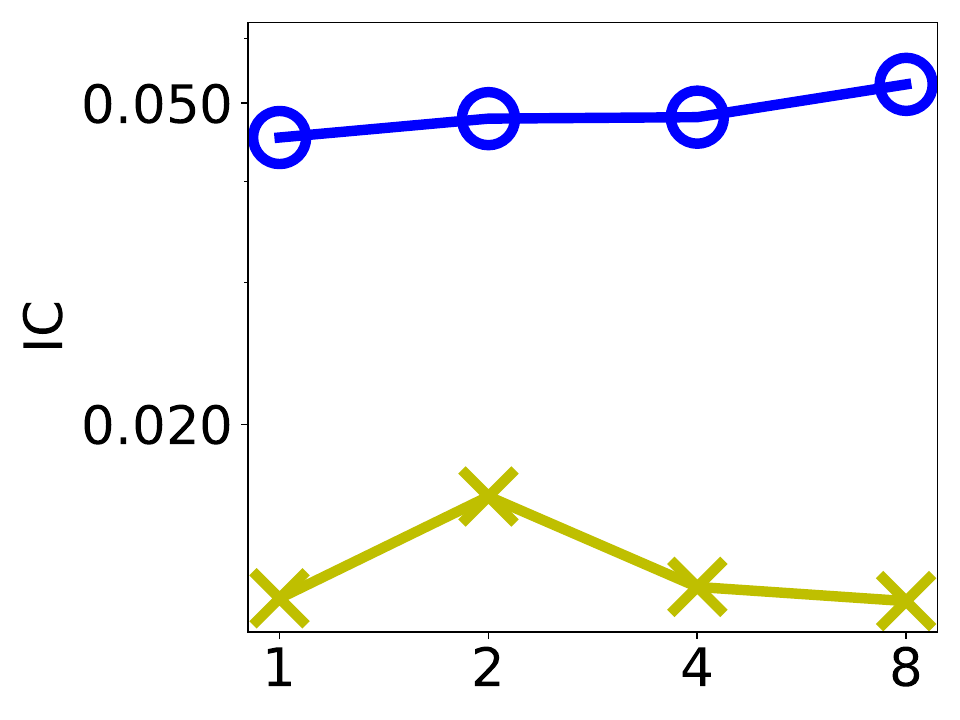}
			\caption{Attention heads}
		\end{subfigure} \hfill
		\begin{subfigure}{0.49\textwidth}
			\includegraphics[width=\textwidth]{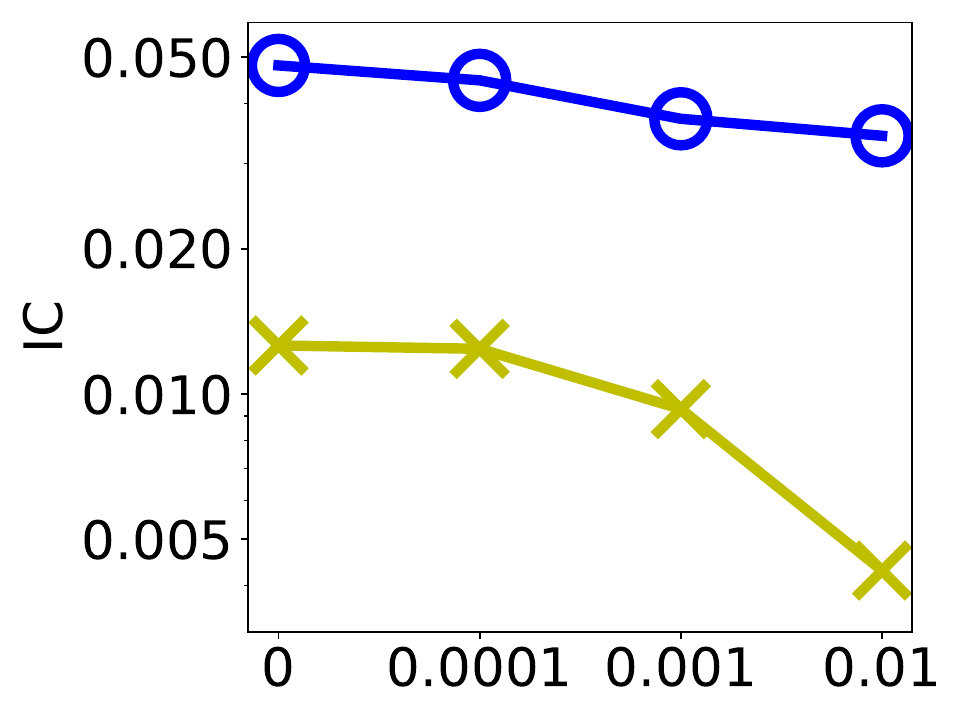}
			\caption{Weight decay}
		\end{subfigure}
	\end{minipage}
	\caption{Comparison of hyperparameter sensitivity between \method and DTML.} \label{figure:hyperparameter}
\end{figure}

\subsection{Hyperparameter Sensitivity (Q3)} \label{hyperparameter}
Figure 4 compares the hyperparameter sensitivity of \method with DTML as the backbone to standalone DTML across four key hyperparameters: (a) hidden layer size, (b) window size, (c) number of attention heads, and (d) weight decay. DTML shows significant sensitivity to these hyperparameters, resulting in unstable performance and fluctuations in IC, especially for window size and weight decay settings.

In contrast, \method consistently maintains high IC values across all hyperparameter settings. This stability is attributed to \method’s ability to mitigate distribution shifts, reducing overfitting and improving generalization.
By aligning distributions across samples, \method decreases the burden on DTML, making the model more robust and less dependent on precise hyperparameter tuning. Consequently, \method achieves more reliable and consistent performance under various training conditions.

\begin{table}
	\centering
	\caption{Empirical evidence that the attention mechanism down-weights noisy observations: (left) Pearson correlation coefficient between $\left|\frac{S_t^c}{S_{t-1}^c}\right|$ and attention weight $\alpha_t$, (right) ratio $\frac{\mathbb{E}\left[\alpha_t \,\middle|\, \left\lvert \frac{S_t^c}{S_{t-1}^c} \right\rvert \ge 0.1\right]}{\mathbb{E}\left[\alpha_t \,\middle|\, \left\lvert \frac{S_t^c}{S_{t-1}^c} \right\rvert < 0.1\right]}$ of attention weights for noisy inputs to normal inputs.}
	\label{tab:attention_noise_summary}
	\begin{tabular}{l|r|r}
		\toprule
		\textbf{Data} & \textbf{Correlation} & \textbf{Attention Ratio} \\
		\midrule
		U.S.    & –0.179 & 0.788 \\
		China   & –0.498 & 0.940 \\
		U.K.    & –0.056 & 0.966 \\
		Korea   & –0.424 & 0.483 \\
		\bottomrule
	\end{tabular}
\end{table}

\subsection{Noise-Aware Behavior of Attention Weights (Q4)} \label{sec:attention_noise}

Table~\ref{tab:attention_noise_summary} shows that \method consistently assigns lower attention weights to time steps with high return magnitudes, indicating its overall tendency to suppress noisy inputs.
Here, noisy inputs refer to observations in which stock prices deviate temporarily from their true distribution due to external shocks, such as unexpected news events or large-scale liquidations.
These results suggest that our attention mechanism effectively identifies and down-weights noisy observations that degrade the estimation of return and volatility.

We quantify this behavior using two complementary analyses.
First, the Pearson correlation coefficient between the absolute return and attention weight is negative across all markets, showing that higher-volatility inputs receive lower attention weights.
Second, the attention ratio between noisy and less noisy time steps is consistently below 1, confirming that noisy observations are down-weighted.

This behavior is attributed to the optimization dynamics of \method.
During training, the attention weights are adjusted to minimize the final loss.
If noisy time steps receive high attention weights, the estimated return and volatility become less accurate, leading to larger prediction errors.
Through backpropagation, the model receives gradients that penalize high attention weights assigned to noisy inputs since these inputs increase prediction error.
Hence, the attention mechanism implicitly filters out noise by learning to trust stable time steps more.

\begin{table}
	\centering
	\caption{
		An ablation study of \method on China dataset.
	}
	\begin{threeparttable}
		\begin{tabular}{l|r|r|r}
			\toprule
			Method &
			IC &  RIC & SR \\
			\midrule
			\method-N &
			0.002  &  -0.011 & 0.158 \\
			\method-E &
			0.045 &  0.040 & 0.595 \\
			\method-D &
			0.035 &  0.020 & 0.168 \\
			\method &
			\textbf{0.048} & \textbf{0.043} & \textbf{2.454} \\
			\bottomrule
		\end{tabular}
	\end{threeparttable}
	\label{table:ablation}
\end{table}

\subsection{Ablation Study (Q5)} \label{ablation}
We analyze the performance of \method with DTML as the backbone, and its variants by removing each of the three primary modules, as shown in Table~\ref{table:ablation}:

\begin{itemize*}
	\item \method-N: \method without the return-volatility normalization
	\item \method-E: \method without the return-volatility estimator
	\item \method-D: \method without the return-volatility denormalization
\end{itemize*}

Each module contributes to improved prediction performance, with \method incorporating the three modules achieving the highest results.
\RVN module contributes the most to the overall performance as it directly addresses the fundamental issue of distribution shift in stock price data.
\RVN effectively reduces distribution discrepancies across samples by normalizing the input data with sample characteristics.
This enables the model to focus on learning normalized patterns from historical error terms. 
This process stabilizes the input distribution, facilitating more robust and efficient feature learning.

\section{Conclusion} \label{conclusion}

We propose \method, an accurate prediction method for alleviating distribution shift in stock price data.
The main ideas of \method are:
1) normalizing price features to address distribution shifts by eliminating individual sample characteristics such as return, volatility, and price scale,
2) precisely estimating these characteristics through an attention-based module to minimize the influence of market anomalies, and
3) reintegrating the sample characteristics into the prediction process to restore lost information.
Additionally, \method leverages geometric Brownian motion to capture long-term trends and neural networks for short-term patterns, effectively integrating their strengths for stock price modeling.
\method boosts the performance of state-of-the-art backbone models in most scenarios, demonstrating notable improvements in predictive accuracy and stability through the proposed normalization approach. 

\section*{GenAI Usage Disclosure} \label{genai}
We have used a generative AI tool (ChatGPT by OpenAI) solely to check for grammatical errors and enhance the clarity of the manuscript during the writing stage.
No part of the research design, experiments, data analysis, or results interpretation was conducted using generative AI.
All technical contributions are the authors’ own work.

\section*{Acknowledgements}
This work was supported by the National Research Foundation of Korea(NRF) funded by MSIT(2022R1A2C3007921).
This work was also supported by Institute of Information \& communications Technology Planning \& Evaluation(IITP) grant funded by the Korea government(MSIT) [No.RS-2021-II211343, Artificial Intelligence Graduate School Program(Seoul National University)], and [NO.RS-2021-II212068, Artificial Intelligence Innovation Hub (Artificial Intelligence Institute, Seoul National University)].
The Institute of Engineering Research at Seoul National University provided research facilities for this work.
The ICT at Seoul National University provides research facilities for this study. U Kang is the corresponding author.
\bibliographystyle{ACM-Reference-Format}
\bibliography{paper}


\begin{thebibliography}{22}

%%% the Web address.

\ifx \showCODEN    \undefined \def \showCODEN     #1{\unskip}     \fi
\ifx \showDOI      \undefined \def \showDOI       #1{#1}\fi
\ifx \showISBNx    \undefined \def \showISBNx     #1{\unskip}     \fi
\ifx \showISBNxiii \undefined \def \showISBNxiii  #1{\unskip}     \fi
\ifx \showISSN     \undefined \def \showISSN      #1{\unskip}     \fi
\ifx \showLCCN     \undefined \def \showLCCN      #1{\unskip}     \fi
\ifx \shownote     \undefined \def \shownote      #1{#1}          \fi
\ifx \showarticletitle \undefined \def \showarticletitle #1{#1}   \fi
\ifx \showURL      \undefined \def \showURL       {\relax}        \fi
% invisible to TeX
\providecommand\bibfield[2]{#2}
\providecommand\bibinfo[2]{#2}
\providecommand\natexlab[1]{#1}
\providecommand\showeprint[2][]{arXiv:#2}

\bibitem[Andersen and Bollerslev(1998)]%
        {andersen1998answering}
\bibfield{author}{\bibinfo{person}{Torben~G Andersen} {and}
  \bibinfo{person}{Tim Bollerslev}.} \bibinfo{year}{1998}\natexlab{}.
\newblock \showarticletitle{Answering the skeptics: Yes, standard volatility
  models do provide accurate forecasts}.
\newblock \bibinfo{journal}{\emph{International economic review}}
  (\bibinfo{year}{1998}), \bibinfo{pages}{885--905}.
\newblock


\bibitem[Andersen et~al\mbox{.}(2001)]%
        {andersen2001distribution}
\bibfield{author}{\bibinfo{person}{Torben~G Andersen}, \bibinfo{person}{Tim
  Bollerslev}, \bibinfo{person}{Francis~X Diebold}, {and}
  \bibinfo{person}{Heiko Ebens}.} \bibinfo{year}{2001}\natexlab{}.
\newblock \showarticletitle{The distribution of realized stock return
  volatility}.
\newblock \bibinfo{journal}{\emph{Journal of financial economics}}
  \bibinfo{volume}{61}, \bibinfo{number}{1} (\bibinfo{year}{2001}),
  \bibinfo{pages}{43--76}.
\newblock


\bibitem[Bjork(2019)]%
        {bjork2019arbitrage}
\bibfield{author}{\bibinfo{person}{Tomas Bjork}.}
  \bibinfo{year}{2019}\natexlab{}.
\newblock \showarticletitle{Arbitrage theory in continuous time}.
\newblock  (\bibinfo{year}{2019}).
\newblock


\bibitem[Cho et~al\mbox{.}(2014)]%
        {cho2014learning}
\bibfield{author}{\bibinfo{person}{Kyunghyun Cho}, \bibinfo{person}{Bart
  Van~Merri{\"e}nboer}, \bibinfo{person}{Caglar Gulcehre},
  \bibinfo{person}{Dzmitry Bahdanau}, \bibinfo{person}{Fethi Bougares},
  \bibinfo{person}{Holger Schwenk}, {and} \bibinfo{person}{Yoshua Bengio}.}
  \bibinfo{year}{2014}\natexlab{}.
\newblock \showarticletitle{Learning phrase representations using RNN
  encoder-decoder for statistical machine translation}.
\newblock \bibinfo{journal}{\emph{arXiv preprint arXiv:1406.1078}}
  (\bibinfo{year}{2014}).
\newblock


\bibitem[Fan et~al\mbox{.}(2023)]%
        {fan2023dish}
\bibfield{author}{\bibinfo{person}{Wei Fan}, \bibinfo{person}{Pengyang Wang},
  \bibinfo{person}{Dongkun Wang}, \bibinfo{person}{Dongjie Wang},
  \bibinfo{person}{Yuanchun Zhou}, {and} \bibinfo{person}{Yanjie Fu}.}
  \bibinfo{year}{2023}\natexlab{}.
\newblock \showarticletitle{Dish-ts: a general paradigm for alleviating
  distribution shift in time series forecasting}. In
  \bibinfo{booktitle}{\emph{Proceedings of the AAAI conference on artificial
  intelligence}}, Vol.~\bibinfo{volume}{37}. \bibinfo{pages}{7522--7529}.
\newblock


\bibitem[Feng et~al\mbox{.}(2019)]%
        {feng2019enhancing}
\bibfield{author}{\bibinfo{person}{Fuli Feng}, \bibinfo{person}{Huimin Chen},
  \bibinfo{person}{Xiangnan He}, \bibinfo{person}{Jie Ding},
  \bibinfo{person}{Maosong Sun}, {and} \bibinfo{person}{Tat-Seng Chua}.}
  \bibinfo{year}{2019}\natexlab{}.
\newblock \showarticletitle{Enhancing Stock Movement Prediction with
  Adversarial Training.}. In \bibinfo{booktitle}{\emph{IJCAI}},
  Vol.~\bibinfo{volume}{19}. \bibinfo{pages}{5843--5849}.
\newblock


\bibitem[Kim et~al\mbox{.}(2024)]%
        {kim2024accurate}
\bibfield{author}{\bibinfo{person}{JinGee Kim}, \bibinfo{person}{Yong-Chan
  Park}, \bibinfo{person}{Jaemin Hong}, {and} \bibinfo{person}{U Kang}.}
  \bibinfo{year}{2024}\natexlab{}.
\newblock \showarticletitle{Accurate Stock Movement Prediction via Multi-Scale
  and Multi-Domain Modeling}. In \bibinfo{booktitle}{\emph{2024 IEEE
  International Conference on Big Data (BigData)}}. IEEE,
  \bibinfo{pages}{1795--1803}.
\newblock


\bibitem[Kim et~al\mbox{.}(2021)]%
        {kim2021reversible}
\bibfield{author}{\bibinfo{person}{Taesung Kim}, \bibinfo{person}{Jinhee Kim},
  \bibinfo{person}{Yunwon Tae}, \bibinfo{person}{Cheonbok Park},
  \bibinfo{person}{Jang-Ho Choi}, {and} \bibinfo{person}{Jaegul Choo}.}
  \bibinfo{year}{2021}\natexlab{}.
\newblock \showarticletitle{Reversible Instance Normalization for Accurate
  Time-Series Forecasting against Distribution Shift}. In
  \bibinfo{booktitle}{\emph{International Conference on Learning
  Representations}}.
\newblock
\urldef\tempurl%
\url{https://openreview.net/forum?id=cGDAkQo1C0p}
\showURL{%
\tempurl}


\bibitem[Kingma(2014)]%
        {kingma2014adam}
\bibfield{author}{\bibinfo{person}{Diederik~P Kingma}.}
  \bibinfo{year}{2014}\natexlab{}.
\newblock \showarticletitle{Adam: A method for stochastic optimization}.
\newblock \bibinfo{journal}{\emph{arXiv preprint arXiv:1412.6980}}
  (\bibinfo{year}{2014}).
\newblock


\bibitem[Li et~al\mbox{.}(2018)]%
        {li2018stock}
\bibfield{author}{\bibinfo{person}{Hao Li}, \bibinfo{person}{Yanyan Shen},
  {and} \bibinfo{person}{Yanmin Zhu}.} \bibinfo{year}{2018}\natexlab{}.
\newblock \showarticletitle{Stock price prediction using attention-based
  multi-input LSTM}. In \bibinfo{booktitle}{\emph{Asian conference on machine
  learning}}. PMLR, \bibinfo{pages}{454--469}.
\newblock


\bibitem[Li et~al\mbox{.}(2024)]%
        {li2024master}
\bibfield{author}{\bibinfo{person}{Tong Li}, \bibinfo{person}{Zhaoyang Liu},
  \bibinfo{person}{Yanyan Shen}, \bibinfo{person}{Xue Wang},
  \bibinfo{person}{Haokun Chen}, {and} \bibinfo{person}{Sen Huang}.}
  \bibinfo{year}{2024}\natexlab{}.
\newblock \showarticletitle{MASTER: Market-Guided Stock Transformer for Stock
  Price Forecasting}. In \bibinfo{booktitle}{\emph{Proceedings of the AAAI
  Conference on Artificial Intelligence}}, Vol.~\bibinfo{volume}{38}.
  \bibinfo{pages}{162--170}.
\newblock


\bibitem[Li et~al\mbox{.}(2021)]%
        {li2021modeling}
\bibfield{author}{\bibinfo{person}{Wei Li}, \bibinfo{person}{Ruihan Bao},
  \bibinfo{person}{Keiko Harimoto}, \bibinfo{person}{Deli Chen},
  \bibinfo{person}{Jingjing Xu}, {and} \bibinfo{person}{Qi Su}.}
  \bibinfo{year}{2021}\natexlab{}.
\newblock \showarticletitle{Modeling the stock relation with graph network for
  overnight stock movement prediction}. In
  \bibinfo{booktitle}{\emph{Proceedings of the twenty-ninth international
  conference on international joint conferences on artificial intelligence}}.
  \bibinfo{pages}{4541--4547}.
\newblock


\bibitem[Lin et~al\mbox{.}(2021)]%
        {lin2021learning}
\bibfield{author}{\bibinfo{person}{Hengxu Lin}, \bibinfo{person}{Dong Zhou},
  \bibinfo{person}{Weiqing Liu}, {and} \bibinfo{person}{Jiang Bian}.}
  \bibinfo{year}{2021}\natexlab{}.
\newblock \showarticletitle{Learning multiple stock trading patterns with
  temporal routing adaptor and optimal transport}. In
  \bibinfo{booktitle}{\emph{Proceedings of the 27th ACM SIGKDD conference on
  knowledge discovery \& data mining}}. \bibinfo{pages}{1017--1026}.
\newblock


\bibitem[Nelson et~al\mbox{.}(2017)]%
        {nelson2017stock}
\bibfield{author}{\bibinfo{person}{David~MQ Nelson},
  \bibinfo{person}{Adriano~CM Pereira}, {and} \bibinfo{person}{Renato~A
  De~Oliveira}.} \bibinfo{year}{2017}\natexlab{}.
\newblock \showarticletitle{Stock market's price movement prediction with LSTM
  neural networks}. In \bibinfo{booktitle}{\emph{2017 International joint
  conference on neural networks (IJCNN)}}. Ieee, \bibinfo{pages}{1419--1426}.
\newblock


\bibitem[Qin et~al\mbox{.}(2017)]%
        {10.5555/3172077.3172254}
\bibfield{author}{\bibinfo{person}{Yao Qin}, \bibinfo{person}{Dongjin Song},
  \bibinfo{person}{Haifeng Cheng}, \bibinfo{person}{Wei Cheng},
  \bibinfo{person}{Guofei Jiang}, {and} \bibinfo{person}{Garrison~W.
  Cottrell}.} \bibinfo{year}{2017}\natexlab{}.
\newblock \showarticletitle{A dual-stage attention-based recurrent neural
  network for time series prediction}. In \bibinfo{booktitle}{\emph{Proceedings
  of the 26th International Joint Conference on Artificial Intelligence}}
  (Melbourne, Australia) \emph{(\bibinfo{series}{IJCAI'17})}.
  \bibinfo{publisher}{AAAI Press}, \bibinfo{pages}{2627–2633}.
\newblock
\showISBNx{9780999241103}


\bibitem[Soun et~al\mbox{.}(2022)]%
        {soun2022accurate}
\bibfield{author}{\bibinfo{person}{Yejun Soun}, \bibinfo{person}{Jaemin Yoo},
  \bibinfo{person}{Minyong Cho}, \bibinfo{person}{Jihyeong Jeon}, {and}
  \bibinfo{person}{U Kang}.} \bibinfo{year}{2022}\natexlab{}.
\newblock \showarticletitle{Accurate stock movement prediction with
  self-supervised learning from sparse noisy tweets}. In
  \bibinfo{booktitle}{\emph{2022 IEEE International Conference on Big Data (Big
  Data)}}. IEEE, \bibinfo{pages}{1691--1700}.
\newblock


\bibitem[Sunny et~al\mbox{.}(2020)]%
        {sunny2020deep}
\bibfield{author}{\bibinfo{person}{Md~Arif~Istiake Sunny},
  \bibinfo{person}{Mirza Mohd~Shahriar Maswood}, {and}
  \bibinfo{person}{Abdullah~G Alharbi}.} \bibinfo{year}{2020}\natexlab{}.
\newblock \showarticletitle{Deep learning-based stock price prediction using
  LSTM and bi-directional LSTM model}. In \bibinfo{booktitle}{\emph{2020 2nd
  novel intelligent and leading emerging sciences conference (NILES)}}. IEEE,
  \bibinfo{pages}{87--92}.
\newblock


\bibitem[Vaswani et~al\mbox{.}(2017)]%
        {vaswani2017attention}
\bibfield{author}{\bibinfo{person}{Ashish Vaswani}, \bibinfo{person}{Noam
  Shazeer}, \bibinfo{person}{Niki Parmar}, \bibinfo{person}{Jakob Uszkoreit},
  \bibinfo{person}{Llion Jones}, \bibinfo{person}{Aidan~N Gomez},
  \bibinfo{person}{{\L}ukasz Kaiser}, {and} \bibinfo{person}{Illia
  Polosukhin}.} \bibinfo{year}{2017}\natexlab{}.
\newblock \showarticletitle{Attention is all you need}.
\newblock \bibinfo{journal}{\emph{Advances in neural information processing
  systems}}  \bibinfo{volume}{30} (\bibinfo{year}{2017}).
\newblock


\bibitem[Yoo and Kang(2021)]%
        {yoo2021attention}
\bibfield{author}{\bibinfo{person}{Jaemin Yoo} {and} \bibinfo{person}{U Kang}.}
  \bibinfo{year}{2021}\natexlab{}.
\newblock \showarticletitle{Attention-based autoregression for accurate and
  efficient multivariate time series forecasting}. In
  \bibinfo{booktitle}{\emph{Proceedings of the 2021 SIAM International
  Conference on Data Mining (SDM)}}. SIAM, \bibinfo{pages}{531--539}.
\newblock


\bibitem[Yoo et~al\mbox{.}(2021)]%
        {yoo2021accurate}
\bibfield{author}{\bibinfo{person}{Jaemin Yoo}, \bibinfo{person}{Yejun Soun},
  \bibinfo{person}{Yong-chan Park}, {and} \bibinfo{person}{U Kang}.}
  \bibinfo{year}{2021}\natexlab{}.
\newblock \showarticletitle{Accurate multivariate stock movement prediction via
  data-axis transformer with multi-level contexts}. In
  \bibinfo{booktitle}{\emph{Proceedings of the 27th ACM SIGKDD Conference on
  Knowledge Discovery \& Data Mining}}. \bibinfo{pages}{2037--2045}.
\newblock


\bibitem[Zhan et~al\mbox{.}(2024)]%
        {zhan2024meta}
\bibfield{author}{\bibinfo{person}{Donglin Zhan}, \bibinfo{person}{Yusheng
  Dai}, \bibinfo{person}{Yiwei Dong}, \bibinfo{person}{Jinghai He},
  \bibinfo{person}{Zhenyi Wang}, {and} \bibinfo{person}{James Anderson}.}
  \bibinfo{year}{2024}\natexlab{}.
\newblock \showarticletitle{Meta-adaptive stock movement prediction with
  two-stage representation learning}. In \bibinfo{booktitle}{\emph{Proceedings
  of the 2024 SIAM International Conference on Data Mining (SDM)}}. SIAM,
  \bibinfo{pages}{508--516}.
\newblock


\bibitem[Zhao et~al\mbox{.}(2023)]%
        {zhao2023doubleadapt}
\bibfield{author}{\bibinfo{person}{Lifan Zhao}, \bibinfo{person}{Shuming Kong},
  {and} \bibinfo{person}{Yanyan Shen}.} \bibinfo{year}{2023}\natexlab{}.
\newblock \showarticletitle{Doubleadapt: A meta-learning approach to
  incremental learning for stock trend forecasting}. In
  \bibinfo{booktitle}{\emph{Proceedings of the 29th ACM SIGKDD Conference on
  Knowledge Discovery and Data Mining}}. \bibinfo{pages}{3492--3503}.
\newblock


\end{thebibliography}

\end{document}